\documentclass[showpacs,superscriptaddress,notitlepage]{revtex4-1}

\usepackage{amsfonts}
\usepackage{amssymb,amsthm,epsfig}
\usepackage{amsmath}
\usepackage{graphicx,color,adjustbox}
\usepackage{dsfont}
\usepackage[small]{subfigure}
\usepackage{fixmath}

\newtheorem{thm}{Theorem}[section]
\newtheorem{lem}[thm]{Lemma}

\begin{document}

\title{A Poisson Kalman filter for disease surveillance}

\author{Donald Ebeigbe$^{*1}$, Tyrus Berry$^{*2}$, Steven J. Schiff$^{1,3}$, Timothy Sauer$^2$}

\begin{abstract} An optimal filter for Poisson observations is developed as a variant of the traditional Kalman filter.  Poisson distributions are characteristic of infectious diseases, which model the number of patients recorded as presenting each day to a health care system. We develop both a linear and nonlinear (extended) filter. The methods are applied to a case study of neonatal sepsis and postinfectious hydrocephalus in Africa, using parameters estimated from publicly available data. Our approach is applicable to a broad range of disease dynamics, including both noncommunicable and the inherent nonlinearities of communicable infectious diseases and epidemics such as from COVID-19.
\end{abstract}

\maketitle

$^*$Contributed equally.

$^1$Center for Neural Engineering, Department of Engineering Science and Mechanics, The Pennsylvania State University, University Park, PA, USA

$^2$Department of Mathematical Sciences, George Mason University, Fairfax, VA, USA

$^3$Center for Infectious Disease Dynamics, and Departments of Neurosurgery and Physics, The Pennsylvania State University, University Park, PA, USA

\section{Introduction}

There has been significant recent interest in the model-based control of disease, specifically using prevention and treatment as methods of control \cite{control1,control2,control3,control4,LQRpaper}. Such model-based  frameworks have been instrumental in our understanding of the dynamics and control of infectious diseases \cite{Keeling2008}, and strategies for global public health policies \cite{Diggle2019}. Successful applications of mathematical modeling and control depend on accurate determination of system states, where data assimilation methods such as Kalman filters \cite{DiseaseFiltering1} play a crucial role in constraining the model with available data.

 Although Kalman filters typically assume Gaussian distributed observations that have direct functional relationships to the state variables, this is unlikely to be suitable for diseases with low rates of occurrence, as is the case in the early or late stages of the spread of most infectious diseases. While a Poisson distribution with a large rate constant can be well approximated by a Gaussian of the same mean and variance, the approximation breaks down when the rates of occurrences are much smaller~\cite{curtis1975simple}. Even more importantly, the variance of a Poisson observation  changes along with its mean, whereas the mean and variance of a Gaussian are decoupled, and often the variance is assumed to be constant (or at least unrelated to the mean) in the Kalman filtering context. 

In this article, we argue that the standard application of Kalman filtering methods is poorly matched to data available during disease surveillance. 
In particular, the assumption of Gaussian noise-perturbed observations is a clear source of inaccuracy when used to model the arrival of patients at medical facilities. We develop a variant of the Kalman filter that assumes Poisson observations, and show how to modify the traditional Kalman equations to produce an optimal filter.

   For one-dimensional systems, an optimal filter has been previously designed for Poisson observations \cite{PoissonFilter}, but has not been generalized to multivariate systems.  Moreover, in order to summarize the true distribution of the state $\vec x_k$ at time step $k$ given the Poisson observations in \cite{PoissonFilter}, a very large number of variables needed to be stored and recalculated at each step. In fact, the number of variables needed also grows very quickly with $k$ (compared to the Kalman filter where the number of variables tracked is constant in $k$).  Another alternative would be to use a fully Bayesian approach such as a particle filter designed with the Poisson likelihood function.  Such an approach would be guaranteed to estimate the true posterior given a sufficiently large ensemble, but such large ensembles often result in high computational complexity.  Instead, we propose a filter which is very similar to the Kalman filter, but is adapted to the unique statistics of the Poisson observations.  The proposed approach maintains the simplicity and computational efficiency of a Kalman filter by only tracking the mean and covariance of the estimated state.  A related linear filter called the Generalized Kalman Filter (GKF) was introduced in \cite{GKFpoisson}, which employed a fixed observation noise covariance matrix that is optimal among all linear filters that are fixed in time.  In contrast, we will derive the optimal time-varying linear filter, and we will use the state estimate to update the observation noise covariance matrix dynamically.  In fact, we will show that the optimal linear filter for Poisson observations is almost identical to the standard Kalman filter except that the observation noise covariance matrix depends on the state estimate.    
   
 In Section \ref{PKFsection}, we first show that by choosing an appropriate observation map, the standard Kalman filter gives an unbiased estimator for Poisson observations.  This justifies using a Kalman filter in the disease modeling context, as long as the observation map is well chosen.  We then show how to modify the Kalman equations to produce an optimal linear filter in the sense of minimizing the expected squared errors. We prove the optimality of this choice in Appendix A. While the optimal filter nominally requires knowledge of the true state, we show empirically that using the filter estimate of the state gives near-optimal performance. We call this approach the Poisson Kalman Filter (PKF).
  
Recently, Li et al.~\cite{li2020substantial} assimilate Poisson observations to carry out modeling of the coronavirus (COVID-19) epidemic. Their modifications to the traditional Kalman filter are in the same spirit to those proposed here, in that the observation noise covariance matrix $V$ is designed to vary with the data. In this article, we derive the Kalman equations that lead to the optimal linear filter, and prove that the optimal choice for linear dynamics is to set $V_k$  to vary proportionally to the number of predicted cases. Nonlinear extensions of the Kalman filter follow standard strategies of generalizing the linear formulas (e.g. the Extended and Ensemble Kalman filters \cite{Schiff2012}). We  develop a nonlinear Extended PKF (EPKF) in Section \ref{EPKF} suitable for contagious infectious disease.  

We should note that an extended Kalman filter has previously been developed in \cite{PoissonProcess1} for point processes where the observation increments are conditionally Poisson given a stochastic hidden variable.  Related approaches were applied to crime statistics in \cite{PoissonProcess3} and neuronal signals in \cite{PoissonProcess2}. In \cite{PoissonProcess2,PoissonProcess3} they assume that the Poisson rates are functions of a hidden state variable, $\vec x_k$ that evolves according to a Markov model, $\vec x_{k+1}=\vec x_k + \textup{noise}$.  In our approach we allow a larger class of stochastic models $\vec x_{k+1}=f(\vec x_k)+\textup{noise}$ with non-trivial dynamics.  Moreover, the filters in \cite{PoissonProcess2,PoissonProcess3} are derived as a Gaussian approximation to a Bayesian posterior, which leads to a nonlinear filter in \cite{PoissonProcess3}. Instead, we take a novel approach by deriving an optimal linear filter for Poisson observations.  

Our case studies start from compartmental models which are built on the standard SIR model and its variants \cite{Keeling2008}.  The SIR model tracks three variables which represent three populations, susceptible (S), infected (I), and recovered (R).  A key feature for communicable disease is that the rate of increase of the infected is proportional to the product of the susceptible and infected populations, $SI$, a nonlinear interaction term that is motivated by the contagious nature of the diseases being modeled.

 In Section \ref{SIRmodel} below, we introduce an SIR model for noncommunicable diseases and show how to apply the Poisson Kalman filter to track the model from example data from two endemic diseases affecting childhood health in Africa -- neonatal sepsis (NS) and postinfectious hydrocephalus (PIH) -- in Sections \ref{SIRmodel} and \ref{SIRHmodel}. Although many of these infections are noncommunicable, acquired during  birth or from the environment afterwards, there is new evidence supporting a role for communicable viruses \cite{Paulson2020}. To our knowledge, there is no existing computational framework that embodies the interdependent dynamics of NS and PIH.  We show how the use of the PKF and EPKF can fill this need.

We discuss future directions both for more detailed study of NS and PIH, and for further extensions of the filtering for infectious disease epidemics, in Section \ref{futuredirections}.

\section{Data assimilation from Poisson observations}\label{PKFsection}

Estimating the current state of a dynamical system is a critical challenge when applying compartmental modeling to disease forecasting and control.  Data assimilation is a method of estimating the state from a time series of noisy observations.  In particular, for a linear system $\vec{x}_{k+1} = f(\vec{x}_k)$, the Kalman filter \cite{KF} gives the optimal state estimate (minimal variance) and also quantifies the uncertainty in the estimate.  However, the Kalman filter was designed for engineering applications where the observations are assumed to have a direct functional relationship to the state variables, except perturbed by Gaussian noise.  

There are at least three reasons why this assumption fails for typical disease surveillance. First, counts of individuals with a disease are by definition nonnegative, contradicting the Gaussian model for uncertainty. Second, the size of the Gaussian noise is decoupled from the population count, being the same magnitude for low populations as for large populations. Finally, in order for the population to be the observed variable, one would have to make a survey, at each time step $k$, of the entire population to directly observe $I_k$, the number of infected at time $k$.  Since this is an unrealistic proposal, the filtering method needs to be adapted to the type of observations that are practical for disease surveillance. We will refer to this modification of the Kalman filter by the name Poisson Kalman Filter (PKF), which we show to be unbiased and optimal among all linear filters.

We operate under that assumption that the disease population cannot be measured directly.
In fact, a reasonable model for observations of disease cases, for example those presenting at a hospital, is a Poisson process, whose rate is proportional to the infected population.  Assume that at time step $k$, the number of new infected patients $I_k$will be approximated by a Poisson random variable with rate $\lambda_{k,I} = c_I I_k$, where $c_I$ is a proportionality constant. 

In a typical filtering problem we would assume that we are given direct observations, $\vec y_k$, of the form $B\vec x_k +  \vec \nu_k$ where $\vec \nu_k$ are random variables representing observation noise.  However, in the Poisson observation context, we instead observe a pair of independent Poisson random variables with rates given by the components of $B \vec x_k$.  We will denote this type of observation by
\[ \vec y_k \sim \textup{Poisson}(B  \vec x_k) \]
meaning that $(\vec y_k)_i$ is Poisson with rate $(B \vec x_k)_i$.  To be more precise we assume that, conditional to $B\vec  x_k$, the components $(\vec y_k)_i$ are independent Poisson random variables with density function,
\[ P\left((\vec y_k)_i = z \, | \, (B  \vec x_k)_i = \lambda \right) = \frac{\lambda^z}{z!}e^{-\lambda} = \frac{((B \vec x_k)_i)^{(\vec y_k)_i}}{((\vec y_k)_i)!}e^{-(B \vec x_k)_i}. \]
The above conditional density makes it clear that $\vec y_k$ and $\vec x_k$ are not independent.

In the case of direct observations, one typically assumes that $\vec y_k$ splits into a sum of two terms, the first of which has deterministic dependence on $\vec x_k$ and the second of which is independent of $\vec x_k$.  However, for Poisson observations this splitting is not possible.  Despite this irreconcilable dependence between $\vec y$ and $\vec x$ the following Lemma shows that if we appropriately center $\vec y$, namely $\vec y - \mathbb{E}[\vec y \, | \, \vec x]$, the result is not correlated with $\vec x$.  

\begin{lem}\label{lem1} Let $\lambda$ be an arbitrary random variable and let $z$ be a Poisson random variable with rate $\lambda$ so that the conditional density of $z$ is $P(z \, | \, \lambda) = \frac{\lambda^z}{z!}e^{-\lambda}$.  Then $\mathbb{E}[(\lambda - \mathbb{E}[\lambda])(z - \mathbb{E}[z \, | \, \lambda])] = 0$.
\end{lem}
\begin{proof}
We first apply the law of total expectation to compute $\mathbb{E}[z] = \mathbb{E}[\mathbb{E}[z\, | \, \lambda]] = \mathbb{E}[\lambda]$ since $\lambda$ is the expected value of a Poisson random variable with known rate $\lambda$.
We then apply the law of total expectation,
\begin{align} \mathbb{E}[(\lambda - \mathbb{E}[\lambda])(z - \mathbb{E}[z \, | \, \lambda]])] &= \mathbb{E}[\mathbb{E}[(\lambda - \mathbb{E}[\lambda])(z - \mathbb{E}[z \, | \, \lambda]]) \, | \, \lambda]]  \nonumber \\
&= \mathbb{E}[(\lambda - \mathbb{E}[\lambda])\mathbb{E}[(z - \mathbb{E}[z \, | \, \lambda]]) \, | \, \lambda]] \nonumber \\
&= \mathbb{E}[(\lambda - \mathbb{E}[\lambda])(\mathbb{E}[z \, | \, \lambda] - \mathbb{E}[z \, | \, \lambda]])] = 0 \nonumber
\end{align}
where the second equality follows from the inner expectation being conditioned on $\lambda$ and the third follows from the linearity of the expectation.
\end{proof}

Lemma \ref{lem1} turns out to be the key to deriving an optimal linear filter for Poisson observations. While Poisson observations are a more realistic model for the type of data available in disease modeling, we now must design a filter which can assimilate this data and produce estimates of the state variable $\vec x_k$.

\subsection{The Poisson Kalman Filter (PKF)}

A linear filter produces an estimate $\hat x_k$ of the true state $\vec x_k$ of the form,
\[ \hat x_k = A_1 \hat x_{k-1} + A_2 \vec y_k \]
where $A_1,A_2$ are matrices.  This is a more restricted class of filters, but we will be able to show that our filter is unbiased, meaning $\mathbb{E}[\hat x_k] = \vec x_k$, and is the optimal linear filter in the sense of giving the minimal squared error.

The PKF assumes a model of the form,
\begin{align}
    \vec x_k &= F \vec x_{k-1} + \vec b_k + \vec \omega_{k-1} \\
    \vec y_k &\sim \textup{Poisson}(B \vec x_k)
\end{align}
where $\vec b_k$ is a known deterministic forcing term, and  $\vec \omega_k$ is dynamical noise with mean zero ($\mathbb{E}[\vec \omega_k] = 0$) and known covariance matrix, $\mathbb{E}[\vec \omega_k \vec \omega_k^\top] = W$.  We also assume that the $\vec \omega_k$ are independent of $\vec x_k, \vec y_k,$ and all other $\vec \omega_\ell$ for $\ell \neq k$. The PKF also assumes that model, $F$, and observation matrices, $B$, are known.  We note that the dynamics $F$ and observation matrix $B$ can also be allowed to change at each step (nonautonomous), but to simplify the notation we assume they are constant.

Like the standard Kalman filter, the PKF is a two-step filter, meaning that it breaks down the estimation of $\hat x_k^+$ from $\hat x_{k-1}^+$ into a forecast step and an assimilation step.  In the forecast step we apply the model to our current estimate $\hat x_{k-1}^+$ to produce the forecast,
\begin{equation}\label{forecast} \hat x_k^- = F \hat x_{k-1}^+ \end{equation}
and in the assimilation step we \emph{assimilate} the new observation by,
\begin{equation}\label{assimilate} \hat x_k^+ = \hat x_k^- + K_k(y_k - B \hat x_k^-). \end{equation}
It is easy to see that this is a linear filter with $A_1 = (I-K_kB)F$ and $A_2 = K_k$.  The filter is defined by the choice of the matrix $K_k$ which is called the \emph{gain matrix}.  Our first result is that any filter of the form \eqref{assimilate} is unbiased.
\begin{thm}\label{unbiased}
Assume that $\mathbb{E}[\hat x_0^+] = \vec x_0$, then for any choice of gain matrices $K_k$ the two step filter defined by \eqref{forecast} and \eqref{assimilate} is unbiased, meaning $\mathbb{E}[\hat x_k^+] =  \vec x_k$.
\end{thm}
The proof of Theorem~\ref{unbiased} is straightforward and can be found in Appendix~\ref{theorem_unbiased}.  The gain matrix is determined by a secondary set of computations which track the covariance matrix, $P_k^+$ for the estimate $\hat x_k^+$.  The covariance matrix is also evolved according to a two step evolution starting with a forecast step,
\[ P_k^- = F P_{k-1}^+ F^\top + W \]
which allows us to calculate the optimal gain matrix,
\begin{equation}\label{gain} K_k = P_k^- B^\top (BP_k^- B^\top + V_k)^{-1} \end{equation}
and then we can complete the assimilation step
\[ P_k^+ = (I - K_kB)P_k^- (I - K_kB)^\top + K_k V_k K_k^\top. \]
While it may seem that $P_k$ is only really necessary in order to compute the gain matrix $K_k$, the matrix $P_k$ also gives an error estimate for the state estimate.

The final component that is required is the $V_k$ matrix in the formula for the optimal gain.  In the standard Kalman filter, $V_k$ is the covariance matrix for the observation noise.  However, in the PKF the variance of the observations is equal to $B \vec x_k$ (meaning $\textup{var}((\vec y_k)_i) = (B \vec x_k)_i$).  So intuitively, we would expect to use $V_k = \textup{diag}(B \vec x_k)$. The next theorem states that this yields the optimal linear filter.

\begin{thm}\label{minimumvariance}
Among all linear filters, the filter given by~\eqref{forecast} and \eqref{assimilate} with gain matrix $K_k$ given by~\eqref{gain} where $V_k = \textup{diag}(B x_k)$ is optimal in the sense of minimal sum of squared errors.  In other words, 
\[ \frac{\partial J_k}{\partial K_k}  = 0\] 
where
\[ J_k =  \textup{trace}(P_k) =  \mathbb{E}[||\hat x_k -  \vec x_k||_2^2] = \sum_i \mathbb{E}[(\hat x_k - \vec x_k)_i^2] \] 
\end{thm}

The proof of Theorem \ref{minimumvariance} is closely related to Lemma \ref{lem1} and can be found in~\ref{minimumvariance_thm}.
Unfortunately, the optimal filter is not accessible since it requires access to the true state $\vec x_k$ in order to define the optimal gain matrix.  Instead, since $\hat x_k$ is an unbiased estimator (for any gain matrix) we approximate the optimal filter by using $V_k = \textup{diag}(B\hat x_{k}^-)$.  We call this approximation the Poisson Kalman Filter (PKF). 

\subsection{PKF Equations}

The discrete-time Poisson Kalman filter (PKF) algorithm is given below.  In order to connect with the potential optimal control applications we include the control term $G_{k-1}\vec u_{k-1}$. If there is no control this term can be dropped.  We also allow all the matrices to vary with time.
 \begin{enumerate}
 \item[1] Dynamical system
 \begin{align}
 \vec x_k &= \max(0,F_{k-1} \vec x_{k-1} +  G_{k-1}\vec u_{k-1} + \vec b_k + \vec w_{k-1}), &   \vec w_k &\sim \mathcal{N}(0,\,W_k)  \nonumber \\
\vec y_k &\sim \textup{Poisson}(B_k \vec x_k)\nonumber \\
 \mathbb{E}[\vec w_k \vec w_j^\top] &= W_k \delta_{k-j} \nonumber \\
 \mathbb{E}[\vec y_k \vec y_j^\top] &= \textup{diag}(B_k \vec x_k)\delta_{k-j} \nonumber \\
 \mathbb{E}[w_k y_j^T] &= 0
 \end{align}
 where $\delta_{k-j}$ is the Kronecker delta function, such that $\delta_{k-j} = 1$ if $k = j$, and $\delta_{k-j} = 0$ if $k \neq j$.  When the state is close to zero the Gaussian noise may move the system into negative values, so at each step we take the maximum of each component and zero. In all the comparisons below, we also apply this maximum to the Kalman filter and extended Kalman filter simulations.  Note that $\textup{diag}(B_k \vec x_k)$ is the true variance of the Poisson observation $\vec y_k$.  However, in the filter below we set $V_k = \textup{diag}(B_k \hat x_k^-)$ since this is the best available estimate.  We now summarize the steps required to obtain the PKF estimates.
 
 \item[2] Initialization
 \begin{eqnarray}
 \hat{x}_0^+ &=& \mathbb{E}[\vec x_0] \nonumber \\
 P_0^+ &=&  \mathbb{E} \left[(\vec x_0 - \hat{x}_0^+)(\vec x_0 - \hat{x}_0^+)^\top \right]
 \end{eqnarray}

 \item[3] Prior estimation (forecast step)
\begin{eqnarray}
\hat{x}_k^- &=& F_{k-1} \hat{x}_{k-1}^+  +  G_{k-1}\vec u_{k-1} + \vec b_k \\
P_k^- &= & F_{k-1} P_{k-1}^+ F_{k-1}^\top + W_{k-1} \\
V_k &=& \textup{diag}(\max(\delta,B_k\hat x_k^-)) \end{eqnarray}

\item[4] Posterior estimation (assimilation step)
\begin{eqnarray}
K_k & = & P_k^- B_k^\top  \left( B_k P_k^- B_k^\top    + V_k \right)^{-1} \\
\hat{x}_k^+ &=& \max\left(0,\hat{x}_k^- + K_k \left(\vec y_k - B_k \hat{x}_k^-   \right) \right) \\
P_k^+ &=&  ( I - K_k B_k) P_k^- ( I - K_k B_k)^\top +  K_k V_k K_k^\top 
\end{eqnarray}
\end{enumerate}
Notice that before the diagonal matrix $V_k$ is formed, we first take the maximum of the diagonal entries and a constant $\delta$.  This is necessary because when the diagonal entries of $V_k$ are too close to zero the filter can become numerically unstable.  The constant $\delta$ should be chosen to be small relative to the average value of the $B_k \vec x_k$, and in all our numerical experiments we set $\delta = 0.1$.
Finally, we note that in practice the initial estimates $\vec x_0^+$ and $P_0^+$ are often not available.  However, the effect of these initial estimates on the accuracy of the state estimates decays to zero exponentially as $k \to \infty$, and often $P_0^+$ is simply chosen to be a multiple of the identity matrix.

\section{An SIR model for noncontagious disease in a restricted population}\label{SIRmodel} 

Severe systemic bacterial infection in the neonatal period, neonatal sepsis (NS), accounts for an estimated 680,000 - 750,000 neonatal deaths per year worldwide \cite{Seale2014} - more than childhood deaths from malaria and HIV combined \cite{Seale2013}. The most common brain disorder in childhood is hydrocephalus, and the largest single cause of hydrocephalus in the world is as a sequelae of NS \cite{kulkarni2017}, accounting for an estimated 160,000 yearly cases of postinfectious hydrocephalus (PIH) in infancy \cite{Dewan2018}. The microbial agents responsible for this enormous loss of human life have been poorly characterized \cite{Saha2018}, although next-generation molecular methods show promise to improve the identification of causal agents \cite{Paulson2020}. Both NS and PIH occur disproportionately in the developing world, and most of the PIH cases will die in childhood without adequate treatment, substantially compounding the effective mortality due to NS and its tremendous burdens on societies \cite{Warf2011EBD,ranjeva2018economic}.

We expect a natural application of the PKF will be to SIR modeling. 
Consider a discrete-time SIR model for neonatal sepsis with three classes: $S_k$ is the susceptible population at time $k$, $I_k$ the infected population, and $R_k$ the recovered population.  (Later, in Section \ref{SIRHmodel}, the model will be expanded to include a postinfectious hydrocephalic class.)  Since there are many unmodeled factors which affect the adult population, and the feedback of neonatal infection into the birth rate takes place on a relatively long time scale, we do not include the adult population in the model.  Thus, $S_k, I_k, R_k$ represent neonatal and infant populations.  Since we are modeling neonatal infections, the susceptible and infected classes are neonatal and, $S_k + I_k$ represents the neonatal population.  The recovered class, $R_k$, will track those that recover from sepsis for a period of time that can be chosen by the modeler as will be described below.  The model is summarized in the diagram in Fig.~\ref{SIRdiagram}.

\begin{figure}
    \centering
    \includegraphics[width=.8\linewidth]{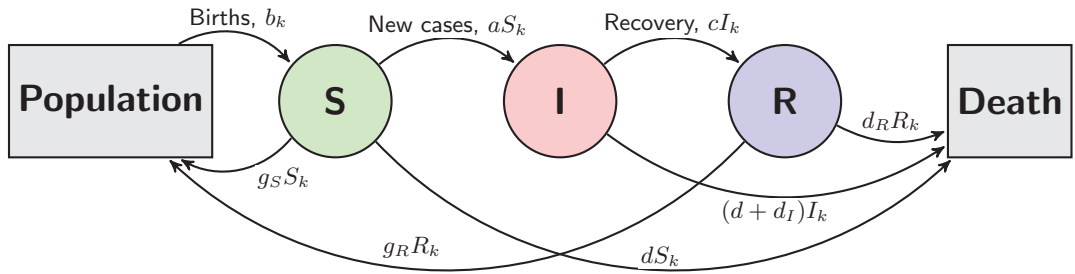}
    \caption{Diagram of the Susceptible-Infected-Recovered (SIR) model for neonatal sepsis.}
    \label{SIRdiagram}
\end{figure}

Modeling only the neonatal/infant populations requires several deviations from the standard SIR model. First, the birth rate is not proportional to any of the model populations, and is instead a forcing, $b_k$, which introduces new population into the susceptible class at each time step.  Moreover, there are now three ways to leave the susceptible class: (1) a neonatal mortality rate $d$, due to factors other than infection (this will affect the two neonatal classes, $S_k$ and $I_k$), (2) an infection rate $a$, which feeds into the infected class, and (3) a `grow-up' rate $g_S$, which signifies no longer being susceptible to neonatal infection. The model is:
\begin{align} 
S_{k+1} &= (1-d-a - g_S)S_k + b_k  \label{S} \\
I_{k+1}  &= (1-d-d_I-c)I_k + aS_k \label{I} \\
R_{k+1} &= (1-d_R - g_R) R_k + cI_k. \label{R}
\end{align}
Notice that the $g_S$ rate removes neonates from the model entirely, so effectively the grow-up rate $g_S$ will control the length of time that we consider to be `neonatal'.  Given a time period $T_S$ for susceptibility, we set $g_S = 1/T_S$, which makes the simplifying assumption that the susceptible population is always equally distributed across different ages.  The grow-up rate $g_R$ controls the length of time that infants in the recovered class are tracked, so that $g_R=1/T_R$ where $T_R$ is the amount of time we track the recovered class.  The two parameters $g_S,g_R$ control the two time scales for susceptibility and recovery (which will become more significant later when we consider the longer time-scale possibility of developing hydrocephalus), and $c$ is the rate of recovery from infection.

With the state variable $\vec x_k = (S_k,I_k,R_k)^\top$, the matrix form of the evolution is
\[ \vec x_{k+1} = F\vec x_k + \vec b_k \]
where
\[ F = \left( \begin{array}{ccc} 1-d-a-g_S & 0 & 0 \\ a & 1-d-d_I-c & 0 \\ 0 & c & 1-d_R-g_R \end{array} \right) \hspace{30pt} \vec b_k = \left( \begin{array}{c} b_k \\ 0 \\ 0 \end{array} \right) \]
If the birth rate is assumed to be constant $b_k \equiv b$, the steady state populations can be explicitly solved.  Setting $S_{\infty} \equiv S_{k+1} = S_k$ in susceptible population in \eqref{S} we can solve for $S_{\infty} = \frac{b}{d+a+g_S}$.  Substituting this for $S_{\infty}=S_k$ in \eqref{I} and setting $I_{\infty}\equiv I_{k+1} = I_k$ in \eqref{I} we can solve for $I_{\infty}$ and similarly we can solve for $R_{\infty}$ giving steady state solutions,
 \begin{align} 
S_{\infty} &= \frac{b}{d+a+g_S} \label{Sinf} \\
I_{\infty}  &= \frac{ab}{(d+d_I+c)(d+a+g_S)} \label{Iinf} \\
R_{\infty} &= \frac{abc}{(d_R + g_R)(d+d_I+c)(d+a+g_S)} \label{Rinf}
\end{align}
These steady state solutions have important public health implications on the time scale where the birth rate is approximately constant.  First, $S_\infty$ determines the scale of public health improvement if susceptibility can be reduced (prevention).  Second, $I_\infty$ determines the resources needed to meet the average infection burden.

\subsection{Case Study: Neonatal sepsis in Uganda}\label{Uganda}

 Publicly available statistics can be used to approximate parameters for NS in Uganda during the time frame 2014-2015.  We consider a discrete time step (the time between steps $k$ and $k+1$) of one day and a neonatal period of $T_S = 28$ days.  From \cite{unicef} we find a 2015 birth rate of $1665000$ per year for Uganda, which for a daily model yields $b\approx 4562$.  Using 2014 statistics for neonatal sepsis in sub-Saharan Africa, we find a neonatal mortality rate of 29 per 1000 with 17\%-29\% attributable to sepsis \cite{ranjeva2018economic}.  For simplicity we assume that the neonatal mortality rate of 29 per 1000 can be divided into 7 attributable to sepsis ($\approx 23\%$ of neonatal mortality, the midpoint of the 17\%-29\% range) and 22 attributable to other causes.  

Since we assume the neonatal period is $T_S$ days, we convert the neonatal mortality rate due to factors other than sepsis into a daily rate by setting $d = 22/1000/T_S$.  The daily neonatal mortality rate due to sepsis is then $7/1000/T_S$, however this is \emph{not} $d_I$ because the $d_I$ variable applies only to the infected class (whereas $d$ applies to both the susceptible and infected classes, and thus is a rate for the entire neonatal population).  That is, $d_I$ represents the daily rate of mortality due to sepsis as a percentage of the population that has sepsis (rather than $7/1000/T_S$ which is the daily rate as a percentage of the entire population).  So before we can determine $d_I$, we first must determine the infection rate $a$.  Infection rate estimates can vary widely based on methodology (\cite{ranjeva2018economic} quotes a range of 5.5 - 170 per 1000 live births). Based on the estimate of one of the authors (SJS) who is a physician conducting medical research on these infants in Uganda, there is a range of 30 - 60 per 1000 live births in that nation.  Conservatively assuming 30 per 1000, we take $a = 30/1000/T_S$ as a daily rate of infection. Now the constant $d_I$ can be determined. We stated above that 7 of the 1000 will die from sepsis, meaning that 7 of the 30 who get sepsis will die from it.  Thus, we find that $d_I = 7/30/T_S$ is the daily rate of death due to sepsis among those that already have sepsis.  This immediately gives us the recovery rate: 7 of the 30 who get sepsis will die from sepsis, and $30(22/1000)$ will die from non-sepsis causes. The remaining $30-7-30(22/1000) = 22.34$ will recover, establishing the recovery rate $c = 22.34/30/T_S$.  Note that 
\[ c = \frac{30 - 7 - 30(22/1000)}{30 T_S} = \frac{1}{T_S} - \frac{7}{30 T_S} - \frac{22/1000}{T_S} = g_S - d_I - d \]
so in fact $c$ is chosen to insure that all of the infected classes leave within the neonatal day period.

The infant mortality rate $m_2$, which covers mortality of the first year after birth, infancy or $T_i$, can also be derived from data. Consider a tracking time for the recovered population of this first year minus the neonatal period, $T_R = T_i-T_S$ (we assume that the recovered population is entirely outside the 28 day neonatal period).  For the death rate in the recovered class we start with the infant mortality rate of 77 per 1000 (in the first year \cite{ranjeva2018economic}) and subtract the 29 per 1000 neonatal mortality rate to find $d_R = 48/1000/T_R$.

\begin{figure}
\subfigure[]{ \boxed{
\begin{minipage}{.51\linewidth}
\vspace{.6em}
\begin{itemize}\setlength\itemsep{.6em}
\item The neonatal time period, $T_S$ (28 days)
\item The infant time period, $T_i$ (365 days)
\item Daily birth rate, $b$ (4562) \cite{unicef}
\item Neonatal mortality rate, $m_1$ (0.0029) \cite{ranjeva2018economic}
\item Percentage of neonatal mortality due to sepsis, $s$ (0.23) \cite{ranjeva2018economic}
\item Infection rate, $a$ (0.0030) \cite{ranjeva2018economic}
\item Infant mortality rate, $m_2$ (0.0077) \cite{ranjeva2018economic}
\end{itemize}
\vspace{.4em}
\end{minipage}}}\hspace{.005\linewidth}
\subfigure[]{\begin{minipage}{.45\linewidth}
    \centering
    \includegraphics[width=.99\linewidth]{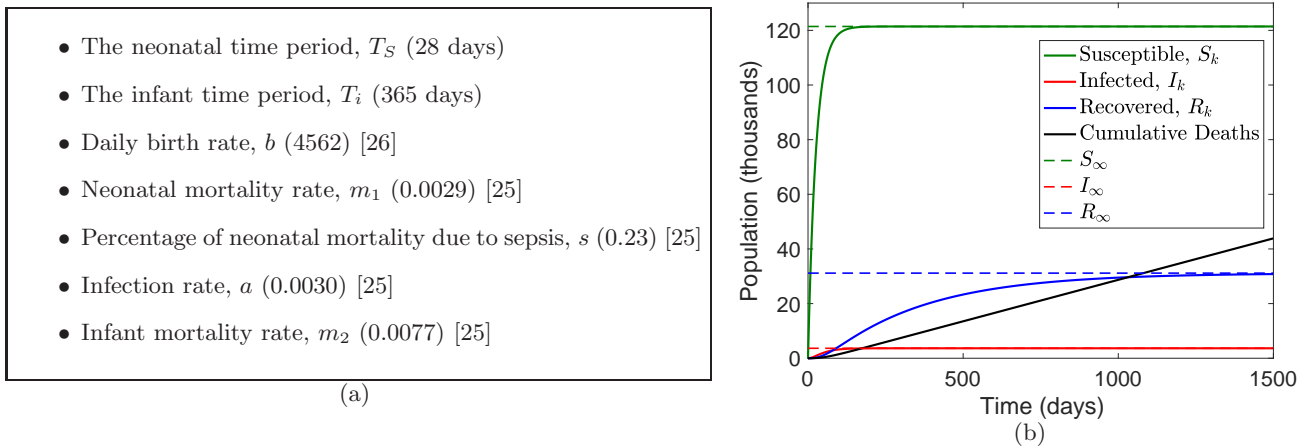}
\end{minipage}}
    \caption{(a) Summary of the inputs to the model for infant sepsis, with values used in parentheses, the remaining parameters are computed using equations \eqref{derivedparams}. (b) Simulation of the model for infant sepsis in Uganda assuming constant birth rate and starting from the zero initial condition, $(S_0,I_0,R_0) = (0,0,0)$.}
    \label{SIRfig}
\end{figure}

We summarize the inputs to the model in Fig.~\ref{SIRfig}
then  compute the parameters $d,d_I,d_R,c,T_R,g_S,g_R$ by
\begin{align}\label{derivedparams}
g_S &= \frac{1}{T_S} &\hspace{20pt} g_R &=\frac{1}{T_R} \nonumber \\
d &= \frac{(1-s)m_1}{T_S} &\hspace{20pt} d_I &=\frac{s m_1}{a T_S}  \\
c &= g_S - d - d_I &\hspace{20pt} d_R &=\frac{m_2 - m_1}{T_R} \nonumber
\end{align}
where $s$ is the fraction of neonatal mortality due to sepsis. The steady state values for the model with these parameters are $S_\infty=121422$, $I_\infty = 3643$, and $R_\infty = 31152$.  We note that the steady state number of infected shows consistency with reported values \cite{unicef}.  
The recovered class is now susceptible to developing PIH.

\section{SIRH: Modeling the hydrocephalic population}\label{SIRHmodel}

We now turn to a model that specifically links neonatal infection and postinfectious hydrocephalus (PIH).  The essential idea is that those that have recovered from sepsis are now susceptible to developing hydrocephalus. The constant $h$ represents the rate at which recovered infants move from the recovered class $R_k$ to a new hydrocephalic class $H_k$, leading to the equations
\begin{align} \label{SIRH}
S_{k+1} &= (1-d-a - g_S)S_k + b_k  \nonumber \\
I_{k+1}  &= (1-d-d_I-c)I_k + aS_k \\
R_{k+1} &= (1-d_R - g_R - h) R_k + cI_k  \nonumber \\
H_{k+1} &= (1- d_R - d_H)H_k + h R_k.  \nonumber
\end{align}
The SIRH system is summarized in the diagram in Fig.~\ref{SIRHdiagram}.

\begin{figure}
    \centering
    \includegraphics[width=.99\linewidth]{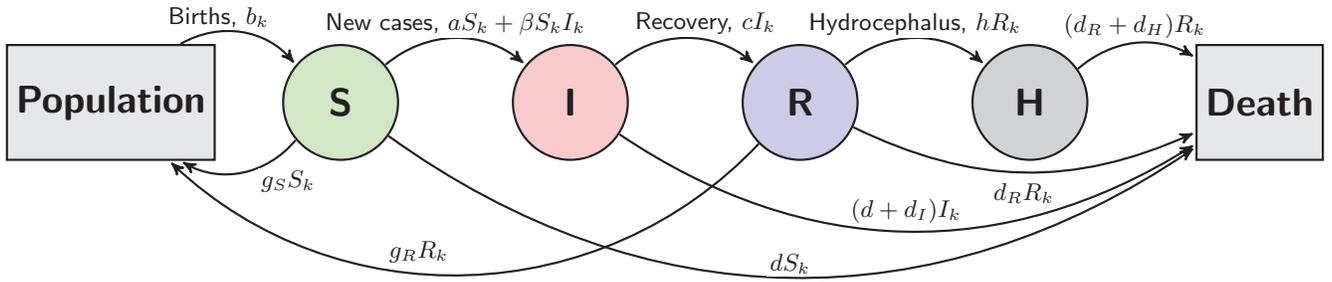}
    \caption{Diagram of the Susceptible-Infected-Recovered-Hydrocephalic (SIRH) model for neonatal sepsis and hydrocephalus. Note that in Section \ref{SIRHmodel} we consider the linear model with $\beta=0$.}
    \label{SIRHdiagram}
\end{figure}
The hydrocephalic class is subject to an additional mortality rate due to hydrocephalus, $d_H$, which requires recalibrating the recovered rate, $d_R$, so that it does not include deaths due to hydrocephalus.  We set $d_R = (m_2 - m_1 - p d_H T_R)/T_R$ where $m_2$ is the infant mortality rate, $m_1$ is the neonatal mortality rate, $p$ is the rate of PIH in the total population under consideration (discussed in Section \ref{PIHstudy} below), and $d_H T_R$ is the rate of death of those who develop PIH during infancy ($d_H$ is the daily rate and $T_R$ is the remainder of the infancy period).  Finally, we note that the steady state value of the recovered class changes from the SIR model due to the rate $h$, and the new steady state along with the hydrocephalic steady state are given by
\begin{align}\label{hinf} 
R_{\infty} &= \frac{abc}{(d_R+g_R+h)(d+d_I+c)(d+a+g_S)} \\
H_{\infty} &= \frac{h R_{\infty}}{d_R + d_H} = \frac{abch}{(d_R+g_R+h)(d+d_I+c)(d+a+g_S)(d_R+d_H)} 
\end{align}
We now return to our case study of modeling PIH in Uganda.

\subsection{Case Study: Infant hydrocephalus in Uganda}\label{PIHstudy}

The first parameter to consider is $h$, the rate of developing postinfectious hydrocephalus (PIH).  In \cite{ranjeva2018economic} it is reported that the incidence of PIH is 3-5 per 1000 live births.  We will take the low estimate of 3 per 1000 setting $p=3/1000$, since it will be shown to be more consistent with other statistics below.  Recall that above we estimated that for 1000 live births there are $30$ cases of sepsis, and $22.34$ of those recover.  Since only recovered sepsis cases can develop PIH, this implies a rate of developing hydrocephalus of \[ h = 3/22.34/T_R.\]   The death rate due to hydrocephalus is highly dependent upon treatment.  The untreated death rate is estimated at 50\%, while treatment can reduce this to 25\%.  We assume an overall death rate of 33\% \cite{warf2011} and we set \[ d_H = 1/3/T_R. \]  Finally, we recalibrate the death rate for those recovering from sepsis by removing the deaths due to hydrocephalus (since those are accounted for in the $H_k$ variable).  So we set 
\[ d_R = \frac{m_2 - m_1 - p\, d_H T_R}{T_R} = \frac{.0077 - .0029 -  .0003 \frac{1}{3}}{T_R} = \frac{.0047}{T_R} \]
The results shown in Fig.~\ref{SIRHfig} predict a steady state of approximately $10000$ ongoing cases of PIH with an annual PIH incidence of approximately 4000 per year ($365*h*R_{\infty}$), and  annual deaths  due to PIH of approximately 3300, consistent with existing estimates \cite{Warf2011EBD}. 

\begin{figure}
    \centering
    \subfigure[]{\includegraphics[width=.4\linewidth]{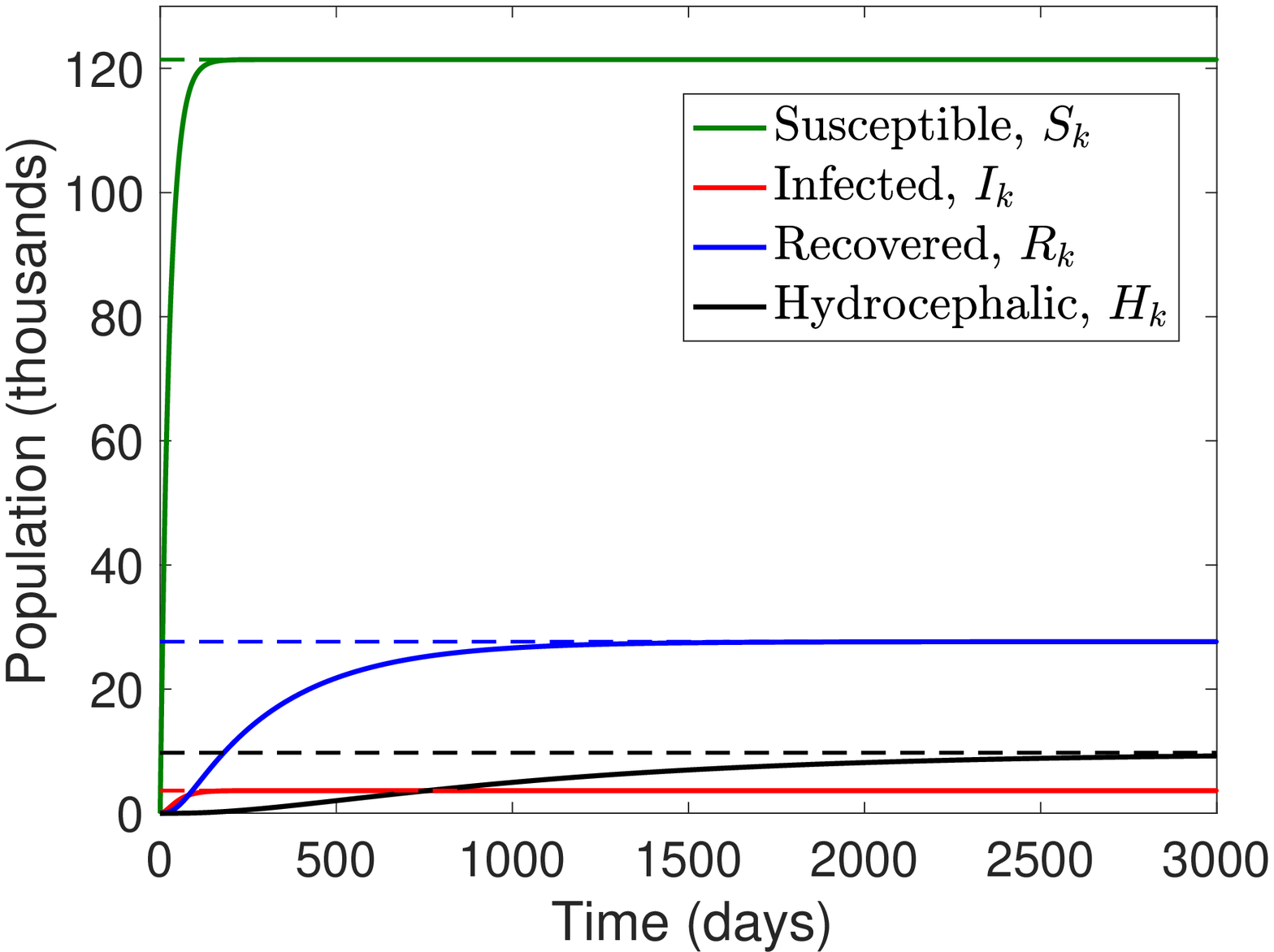}}\subfigure[]{\includegraphics[width=.4\linewidth]{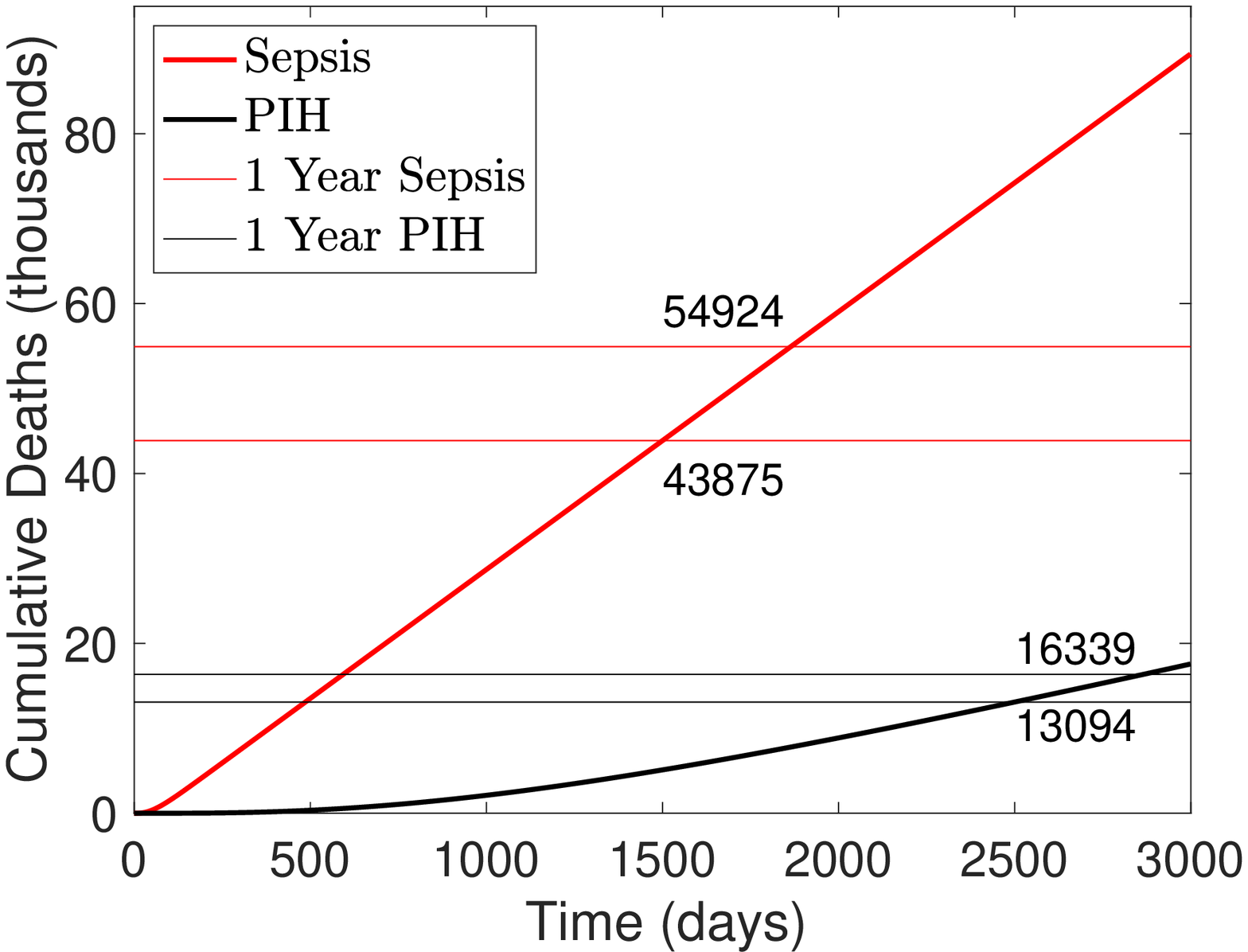}}
    \caption{(a) Simulation of the SIRH model for Uganda starting from the zero initial condition. (b) Plot of the cumulative deaths from sepsis and hydrocephalus in the simulation. The horizontal lines are spaced so that their intersections with the curves are 365 days apart and indicated the cumulative deaths at times one year apart.  The model predicts approximately 11000 annual deaths due to sepsis and approximately 3300 annual deaths due to PIH.}
    \label{SIRHfig}
\end{figure}

\subsection{PKF Simulations}

\begin{figure}
    \centering
    \subfigure[]{\includegraphics[width=.40\linewidth]{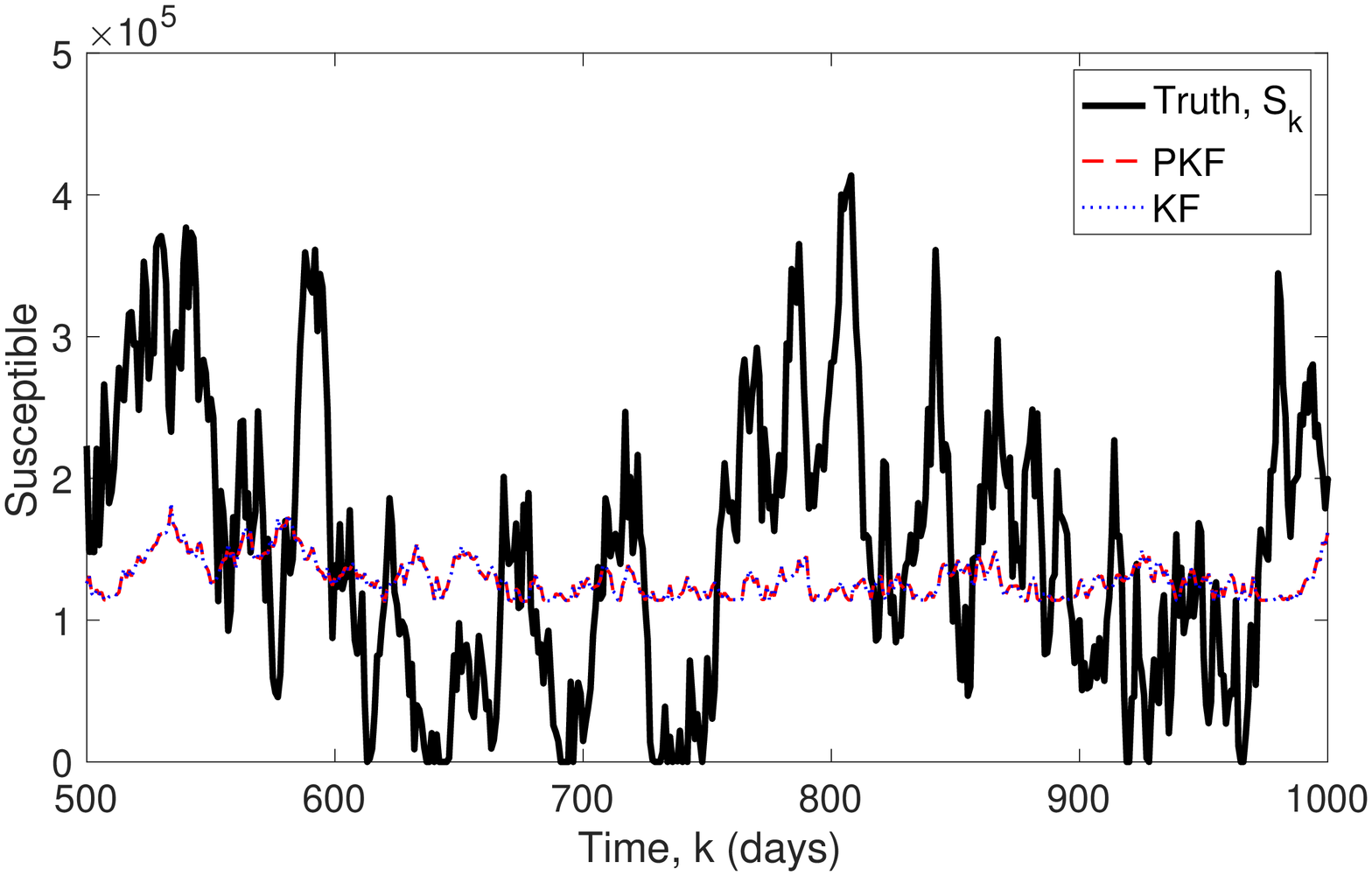}}
    \subfigure[]{\includegraphics[width=.40\linewidth]{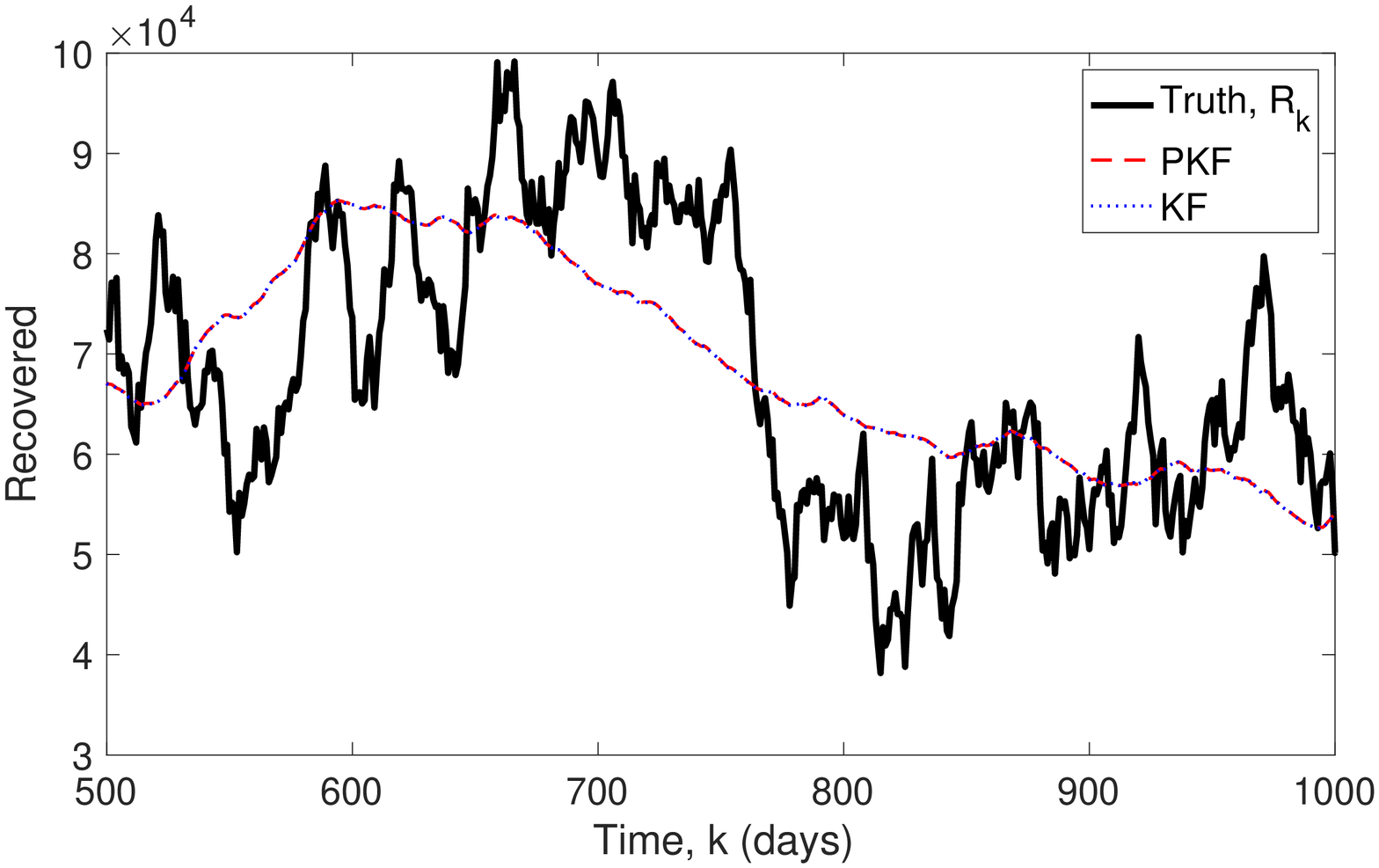}} \\
    \subfigure[]{\includegraphics[width=.40\linewidth]{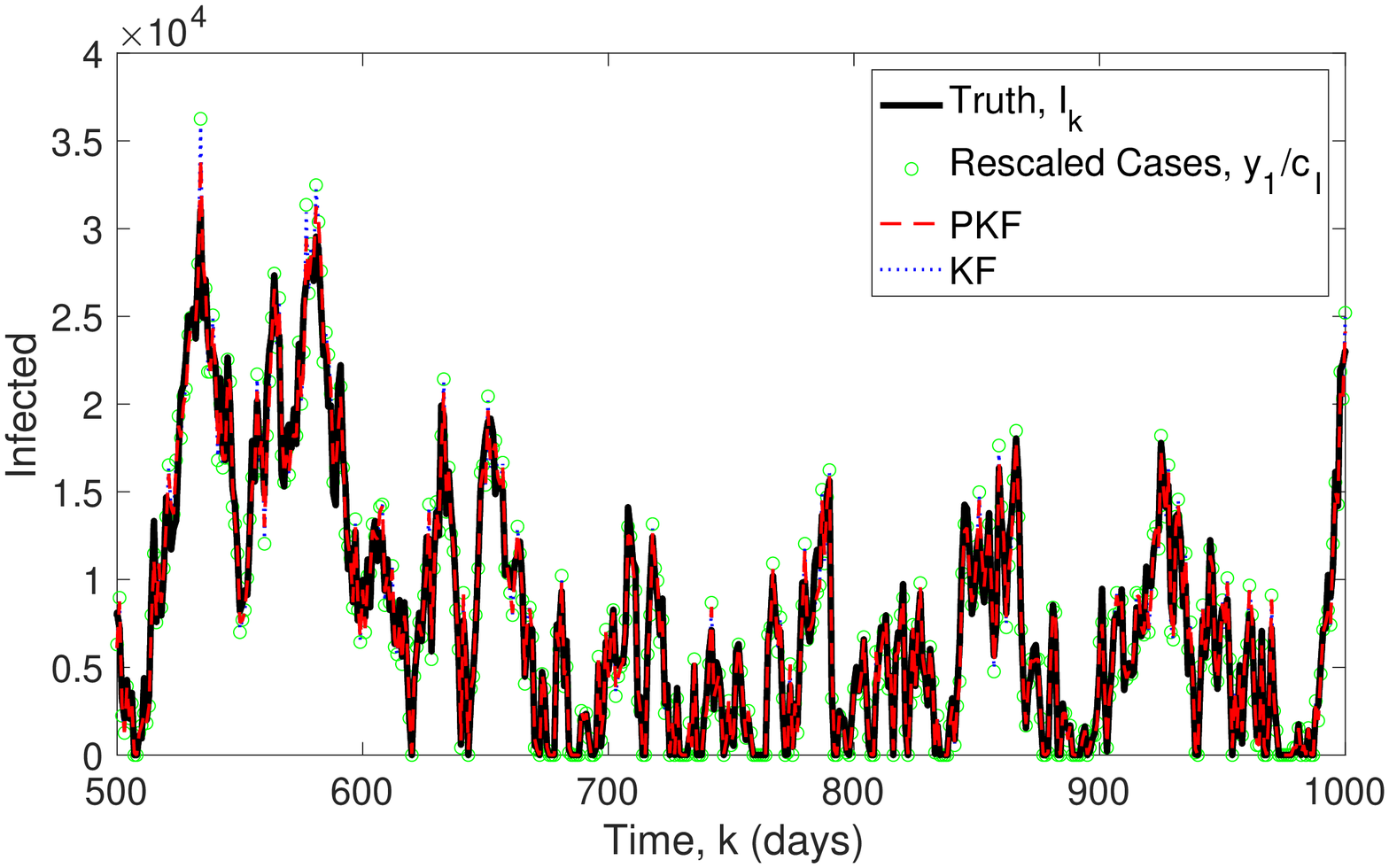}}
    \subfigure[]{\includegraphics[width=.40\linewidth]{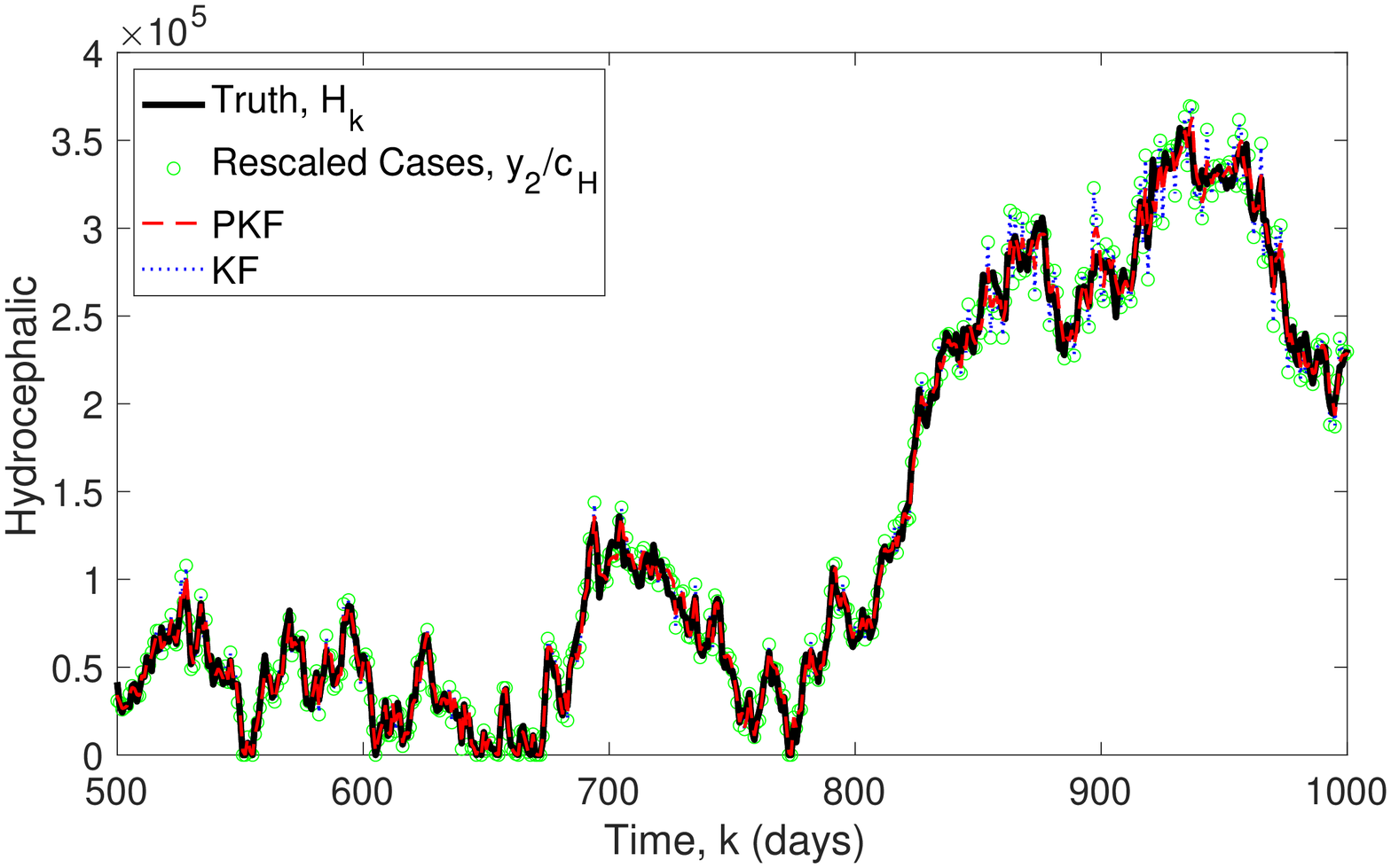}}\\
    \subfigure[]{\includegraphics[width=.40\linewidth]{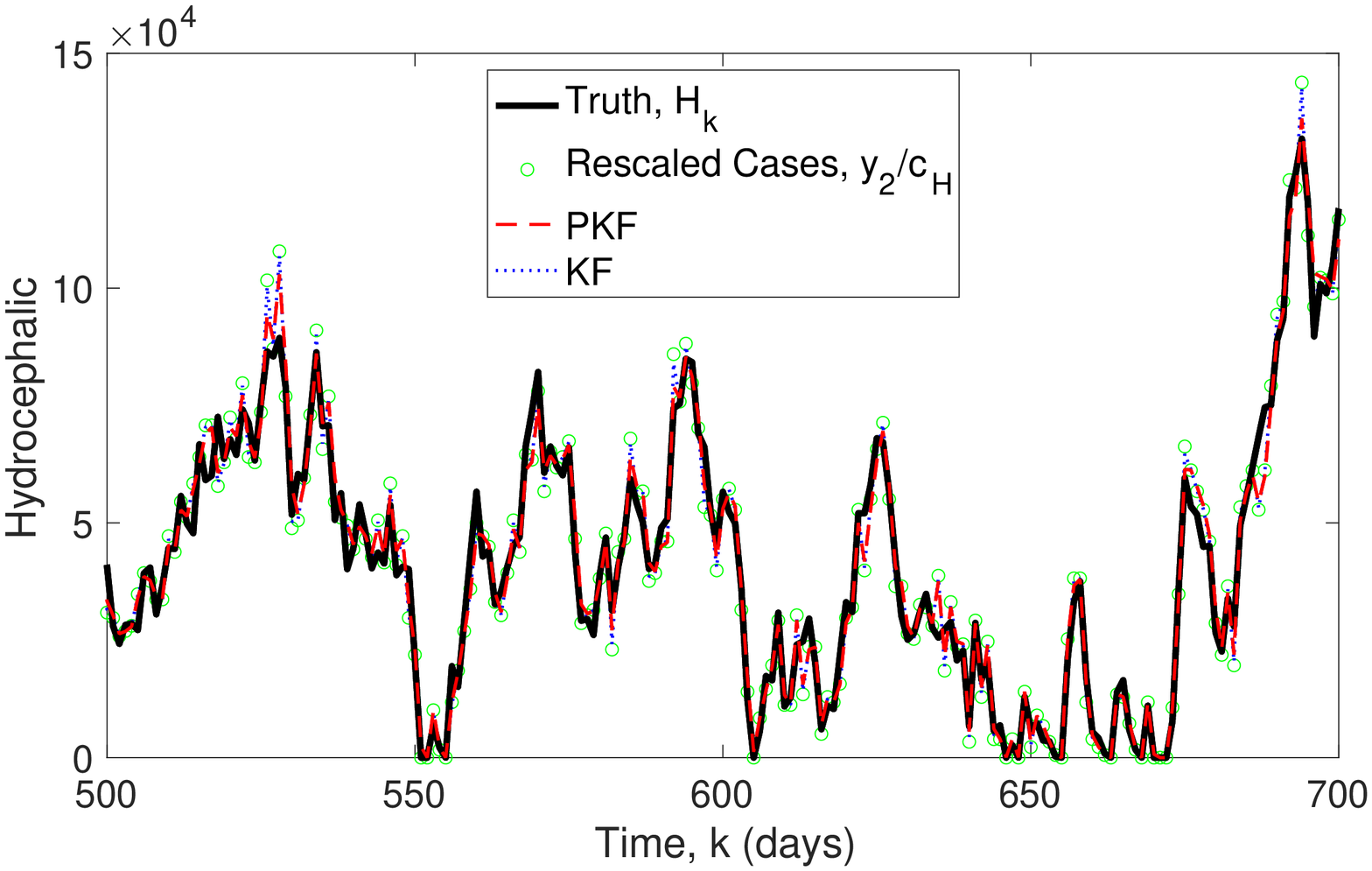}}
    \subfigure[]{\includegraphics[width=.40\linewidth]{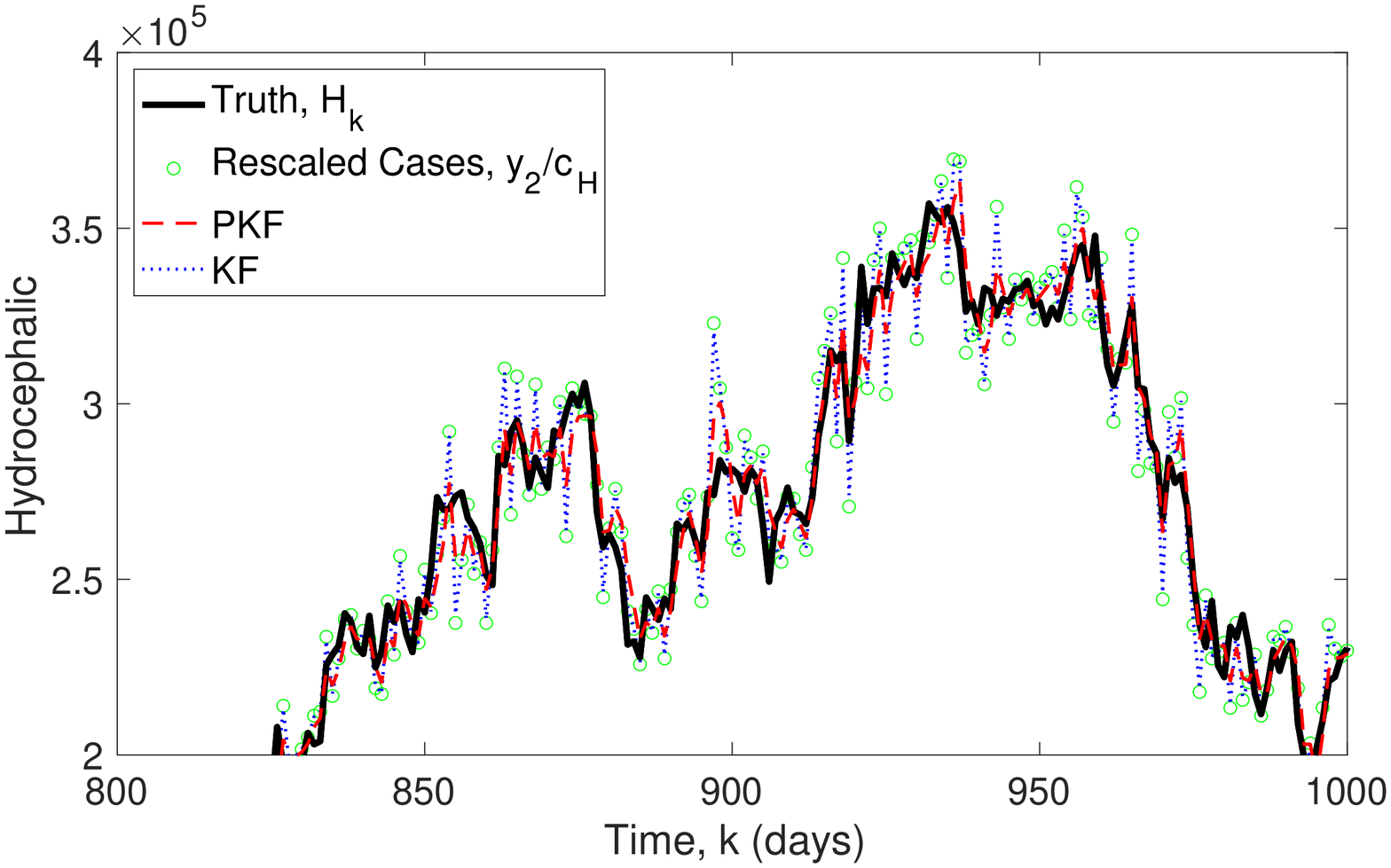}}
    \subfigure[]{\includegraphics[width=.40\linewidth]{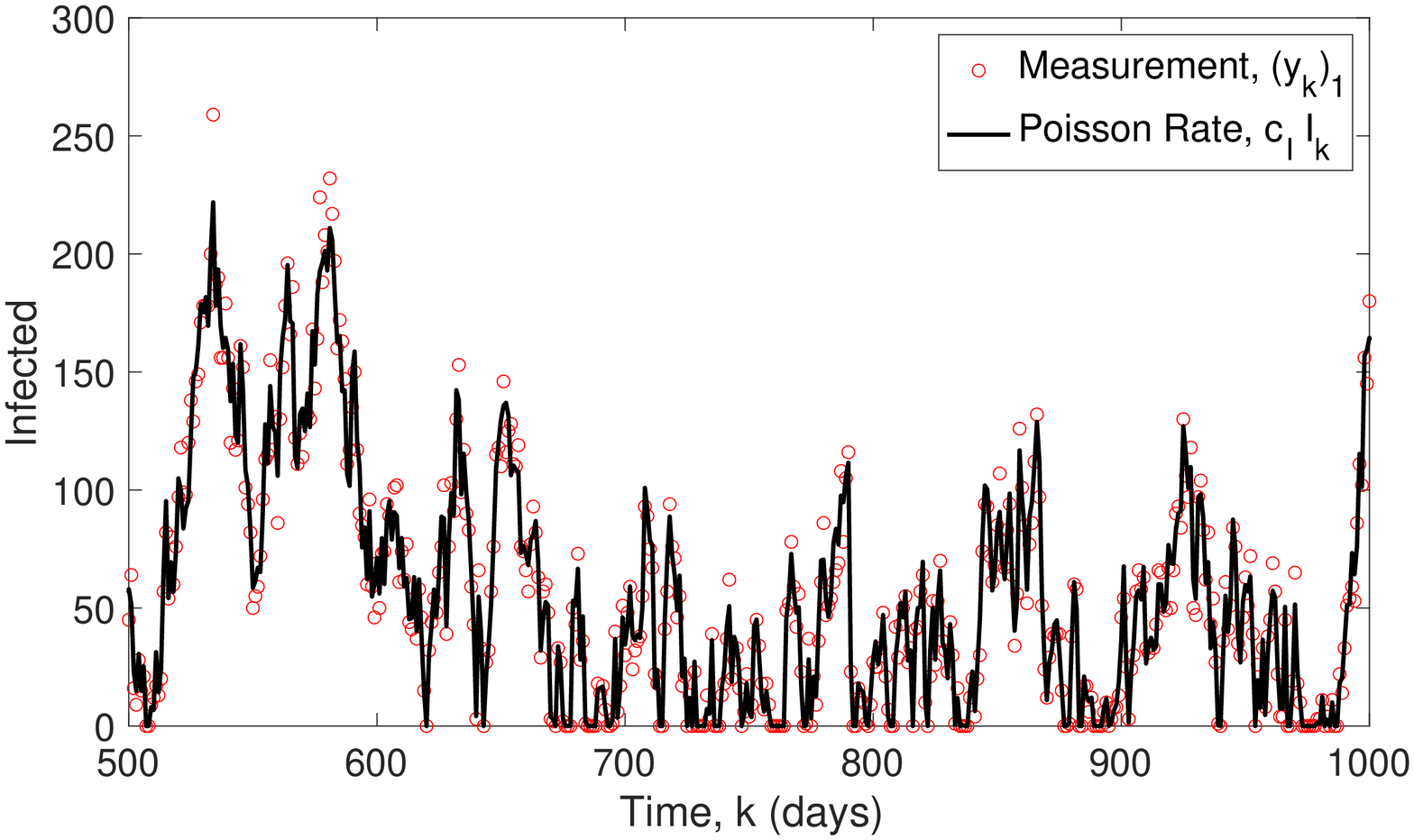}}
    \subfigure[]{\includegraphics[width=.40\linewidth]{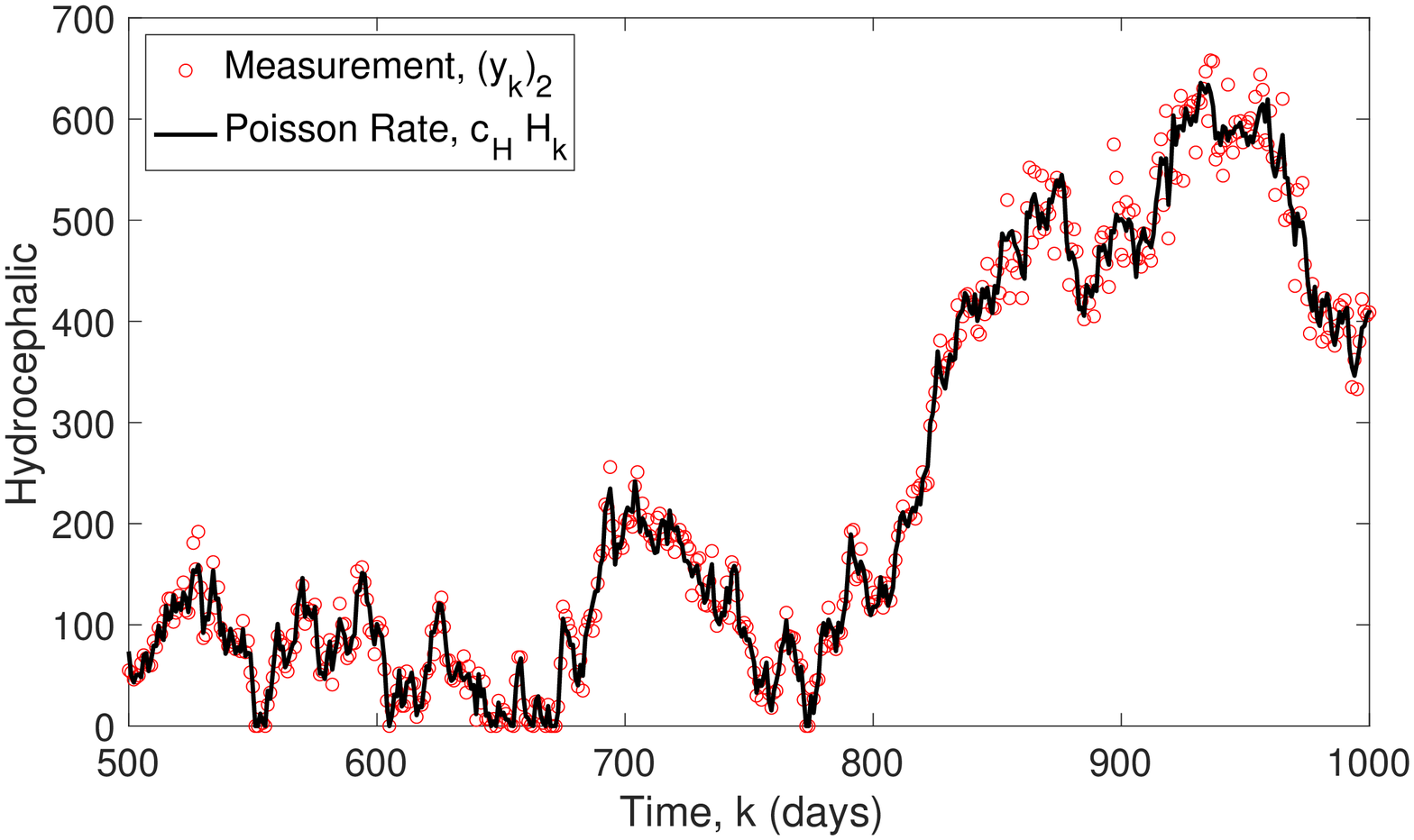}}\vspace{-4pt} 
    \caption{Comparison of the PKF (optimal variable gain) and the Kalman filter (optimal fixed gain) for the SIRH model with Poisson observations of the infected and hydrocephalic populations. (a)-(d) compare the true S, I, R and H values (black) to the PKF (red, dashed) and Kalman filter (blue, dotted) estimates.  Infected and hydrocephalic also show the observations (green, circles) rescaled by dividing by the constants $c_I, c_H$ respectively. (e)-(f) are expanded versions of (d), enlarged to show detail. When the number of cases is large, the KF estimate of H is very close to the observations, whereas the PKF adjusts to the larger observation variance and produces better estimates. (g)-(h) show the Poisson rates (black) of I and H and the observed case numbers (red, circles) from the Poisson distribution.}
    \label{estim1Fig}
\end{figure}

Using the SIRH model described in~Section~\ref{SIRHmodel}, we evaluate the performance of the PKF when the observations follow a random Poisson distribution with known rates $\lambda_{k,I} = c_I I_k$ and $\lambda_{k,H} = c_H H_k$, which represent the number of infants with sepsis and the number of infants with hydrocephalus that show up at the hospital. Our observations can be written as
\[ y_k = \left( \begin{array}{c} y_{k,1} \\ y_{k,2} \end{array} \right) =  \mbox{Poisson}\left( \begin{array}{c} \lambda_{k,I} \\ \lambda_{k,H} \end{array} \right), \quad \mbox{where} \, \left( \begin{array}{c} \lambda_{k,I} \\ \lambda_{k,H} \end{array} \right) =  B x_k = \left( \begin{array}{cccc} 0 & c_I & 0 & 0 \\ 0 & 0 & 0 & c_H \end{array} \right)\left( \begin{array}{c} S_k \\ I_k \\ R_k \\ H_k \end{array} \right). \]
and we start the system at the equilibrium values.

The simulation in Fig.~\ref{estim1Fig} was run with system noise $W = \textup{diag}(144,1,1,10) \times 10^7$ and the constant daily birth rate $b = 4562$, while setting the sepsis and hydrocephalus proportionality constants as $c_I = 0.2/T_S$ and $c_H = 0.6/T_R$ respectively.  The idea behind these values is that if 20\% of total sepsis cases seek care over the entire $T_S$ period of sepsis susceptibility, then the daily rate of arrivals would be $0.2/T_S$ multiplied by the number of true sepsis case (20\% was chosen purely for purposes of simulation).  In Fig.~\ref{estim1Fig} we see that the PKF (red, dashed curves) gave good estimates of the observed variables, namely the, infected and hydrocephalic populations.  The PKF also obtains information about the unobserved variables, namely, the susceptible and recovered populations at least on a slow time scale, however the fast time scale information about the unobserved variables seems limited.   

Fig.~\ref{estim1Fig} also compares the PKF to a Kalman filter (blue, dotted curves) which was given the optimal fixed observation noise covariance matrix, $V_{\textup{const}} = \textup{diag}(B \overline x)$ where $\overline x$ is the time average of the state variables.  The disadvantage of the fixed gain is that when the number of infected or hydrocephalic is large the variance of the observations will be larger than the average value.  This means that the Kalman filter will underestimate the observation variance and use an oversized gain.  This is shown in Fig.~\ref{estim1Fig} where the Kalman filter estimates closely follows the observations when the number of infected or hydrocephalic are large.  The PKF dynamically adjusts the observation covariance matrix based on the state estimate in order to prevent this.  This is further shown in Fig.~\ref{estim1Fig2} which compares the root mean squared error (RMSE) for the PKF and the Kalman filter for various levels of system noise. Fig.~\ref{estim1Fig2} also compares the PKF, which uses the filter estimate to determine $V_k$, to an oracle PKF which uses the true state for $V_k$ and we see that their performance is almost identical even at high noise levels.  In this case (with a linear model), the Kalman filter and PKF have similar performance for the unobserved variables, which seems to indicate that they are relying more on the stability of the model rather than correlations with the observed variables. Of course this is reliant on using the optimal $V_{\rm const}$ matrix in the EKF.  Moreover, in the context of disease surveillance filtering the observed variables is critical to account for over/under reporting in producing a clean data set, and for these variables the PKF has a significant advantage.

Finally, we note that the PKF has the largest advantage at high noise levels.  This is because the SIRH system is a stable linear system, so that noise is the only unstable component of the dynamics.  In the absence of noise, no filter would be necessary since all trajectories would converge to the equilibrium regardless of observations.  This suggests that a generalized PKF (such as the Extended PKF considered below) would have an advantage for nonlinear dynamics with unstable directions even in the absence of system noise.

\begin{figure}
    \centering
    \subfigure[]{\includegraphics[width=.49\linewidth]{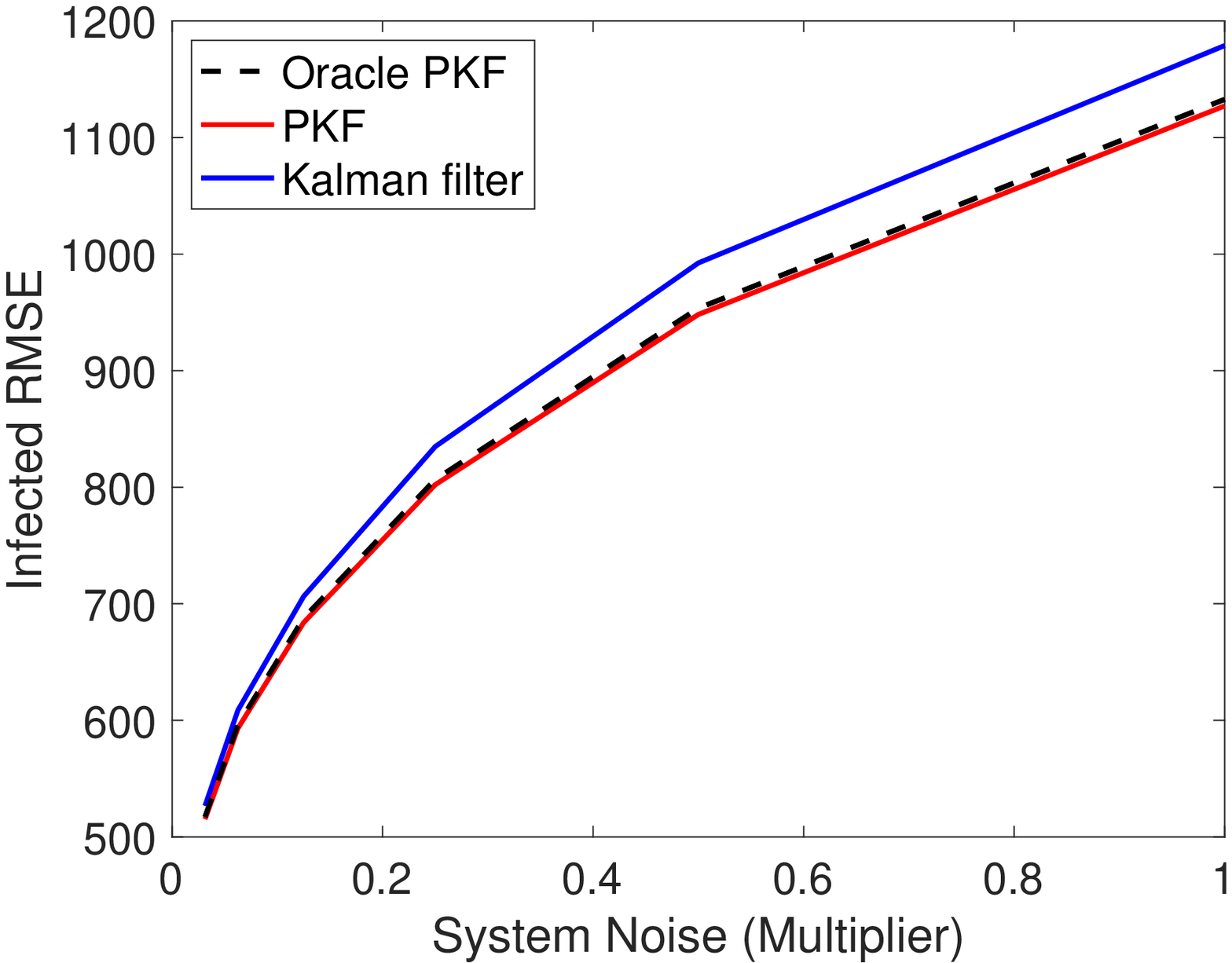}}
    \subfigure[]{\includegraphics[width=.49\linewidth]{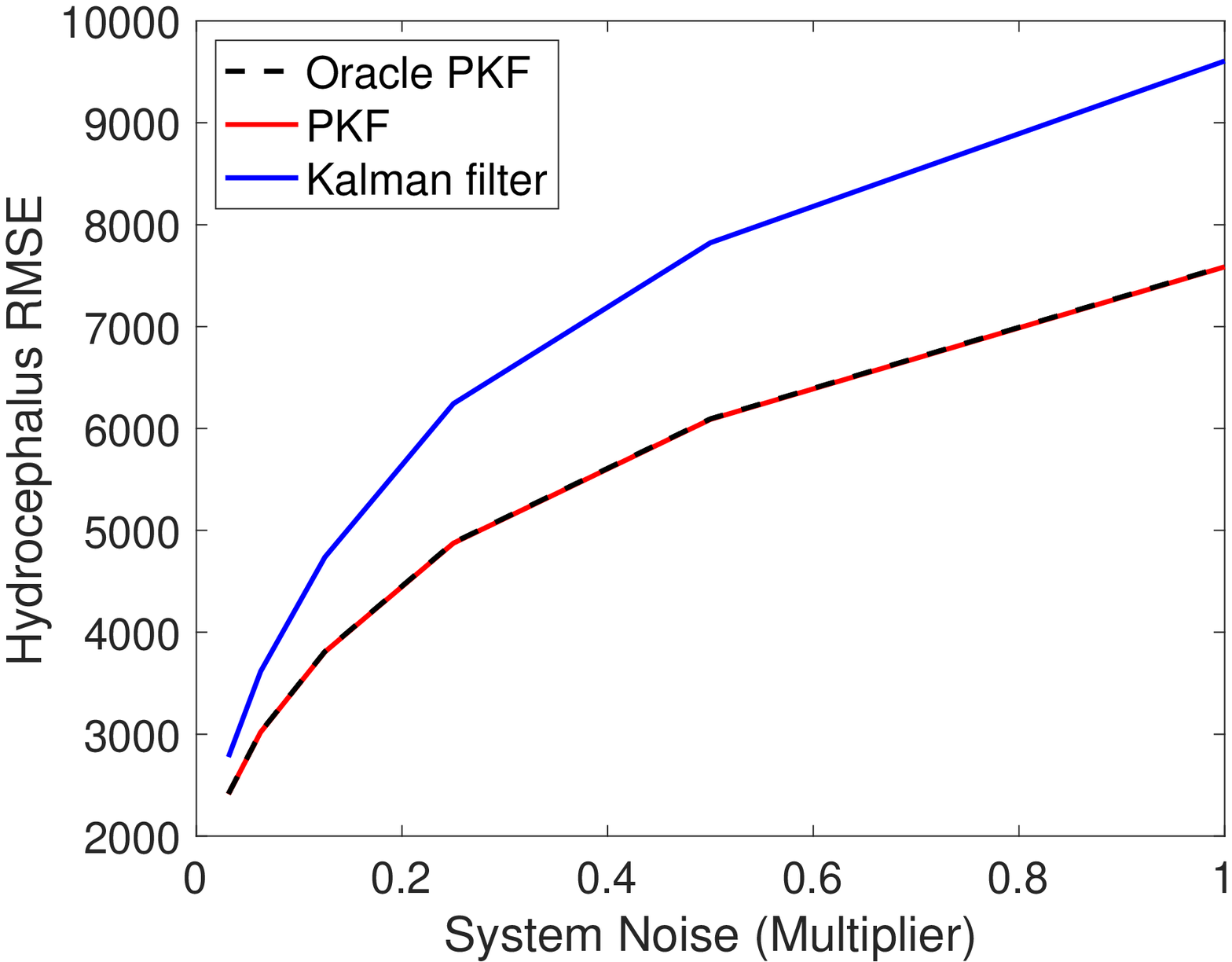}} 
    \caption{Comparison of the RMSE for the infected (a) and hydrocephalic (b) populations of the PKF (red, optimal variable gain) and the Kalman filter (blue, optimal fixed gain) as function of the system noise. We also compare to an oracle PKF (black, dashed) which is given the optimal choice of $V_k = \textup{diag}(B\vec x_k)$.  System noise is quantified as a multiple of the base noise level $W$. The RMSE is averaged over $10^6$ filter steps.}
    \label{estim1Fig2}
\end{figure}


\section{An Extended Poisson Kalman Filter for Contagious Disease}\label{EPKF}

So far we have considered a linear model for NS, which is sufficient for noncontagious infections.  However, contagious disease models typically contain an nonlinearity that models the contagious spread.  In order to broaden the applicability of the PKF we now show that it also offers improvements for these nonlinear models by using a standard approach to extend the Kalman equations to nonlinear dynamics.  Moreover, because there are potential mechanisms for contagious infections contributing to NS \cite{Paulson2020}, modeling these infections requires a nonlinear system.  As in the classical SIR model we assume that the contagious spread will be simultaneously proportional to the both the number of susceptibles and the number of infected and so we model the number of contagious cases at time $k$ as $\beta S_k I_k$, where $\beta$ is infectivity.  Introducing this term to the SIRH model we have
\begin{align} \label{SIRHcontagious}
S_{k+1} &= (1-d-a - g_S)S_k - \beta S_k I_k + b_k  \nonumber \\
I_{k+1}  &= (1-d-d_I-c)I_k + aS_k + \beta S_k I_k \\
R_{k+1} &= (1-d_R - g_R - h) R_k + cI_k  \nonumber \\
H_{k+1} &= (1- d_R - d_H)H_k + h R_k.  \nonumber
\end{align}
The state of the nonlinear model is $\vec x_{k+1} = f_k(\vec x_{k})$ where $\vec x_k = (S_k,I_k,R_k,H_k)$. When the birth rate is constant $b_k = b$ we can write $f_k = f$ and the system can be considered autonomous, but we also allow nonautonomous dynamics as long as each $f_k$ is known.  This model is of significant interest since estimating the $a$ and $\beta$ parameters from data would help determine the role of contagious spread in NS.

\begin{figure}
    \centering
    \subfigure[]{\includegraphics[width=.40\linewidth]{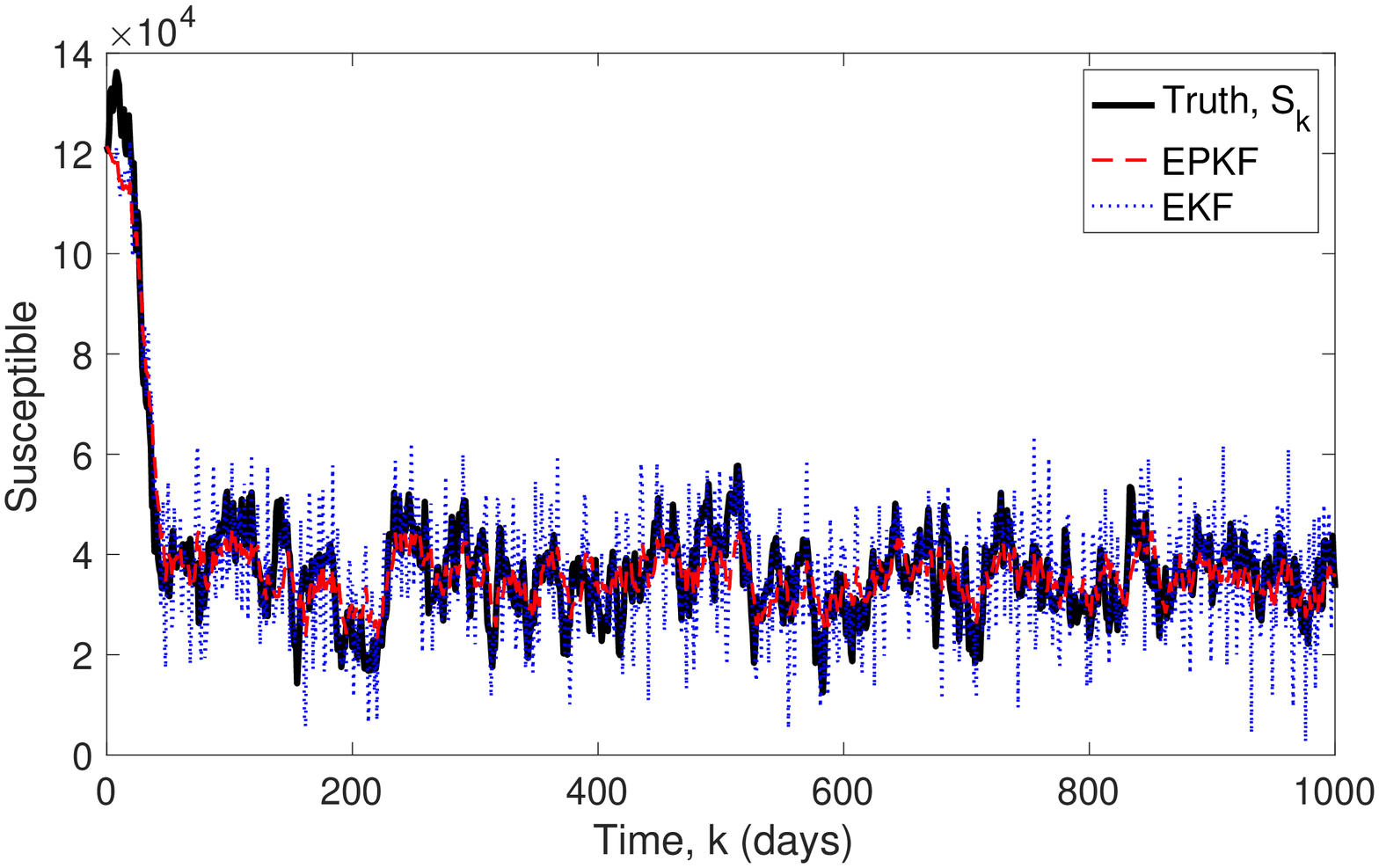}}
    \subfigure[]{\includegraphics[width=.40\linewidth]{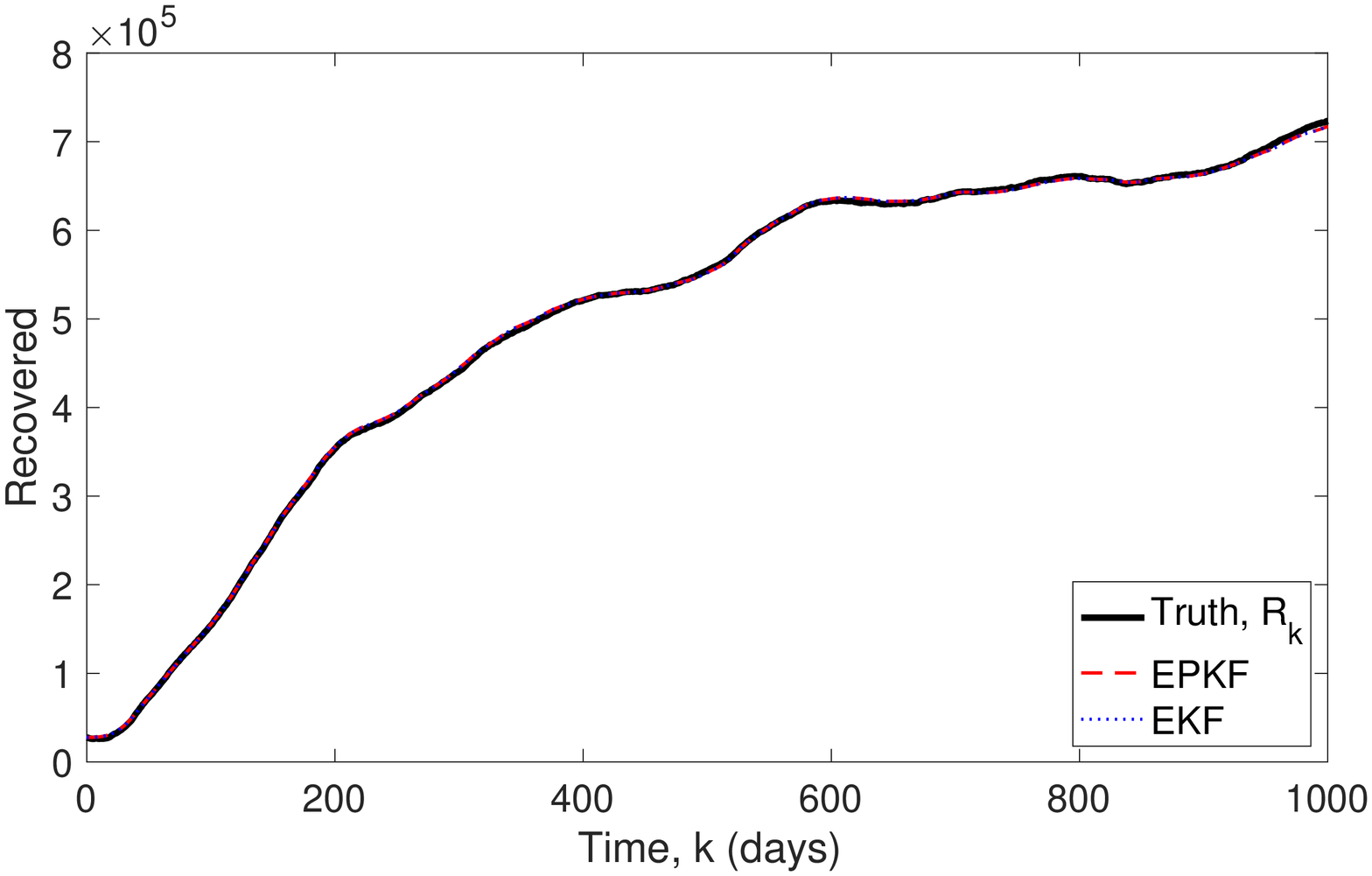}} \vspace{-4pt} \\
    \subfigure[]{\includegraphics[width=.40\linewidth]{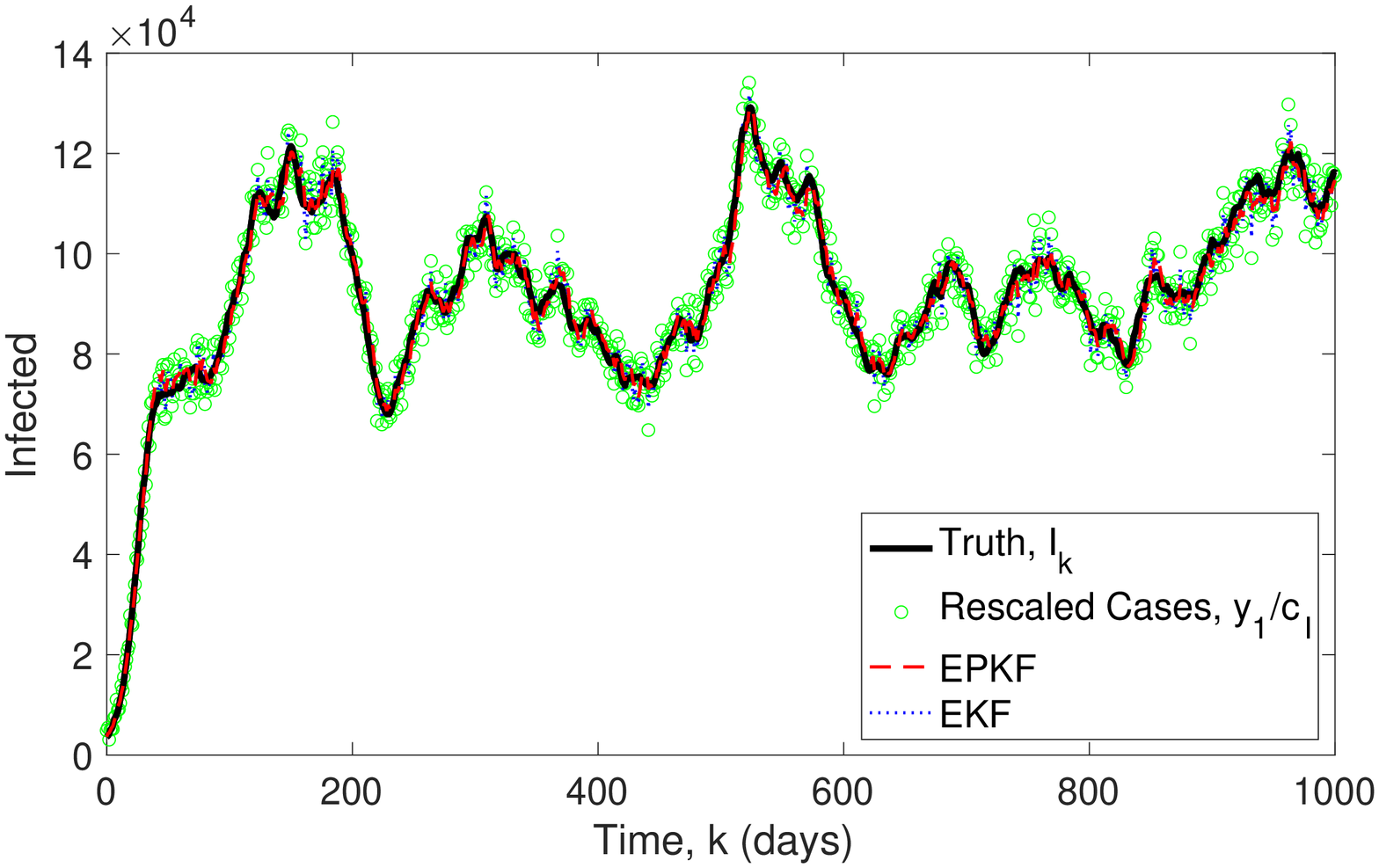}}
    \subfigure[]{\includegraphics[width=.40\linewidth]{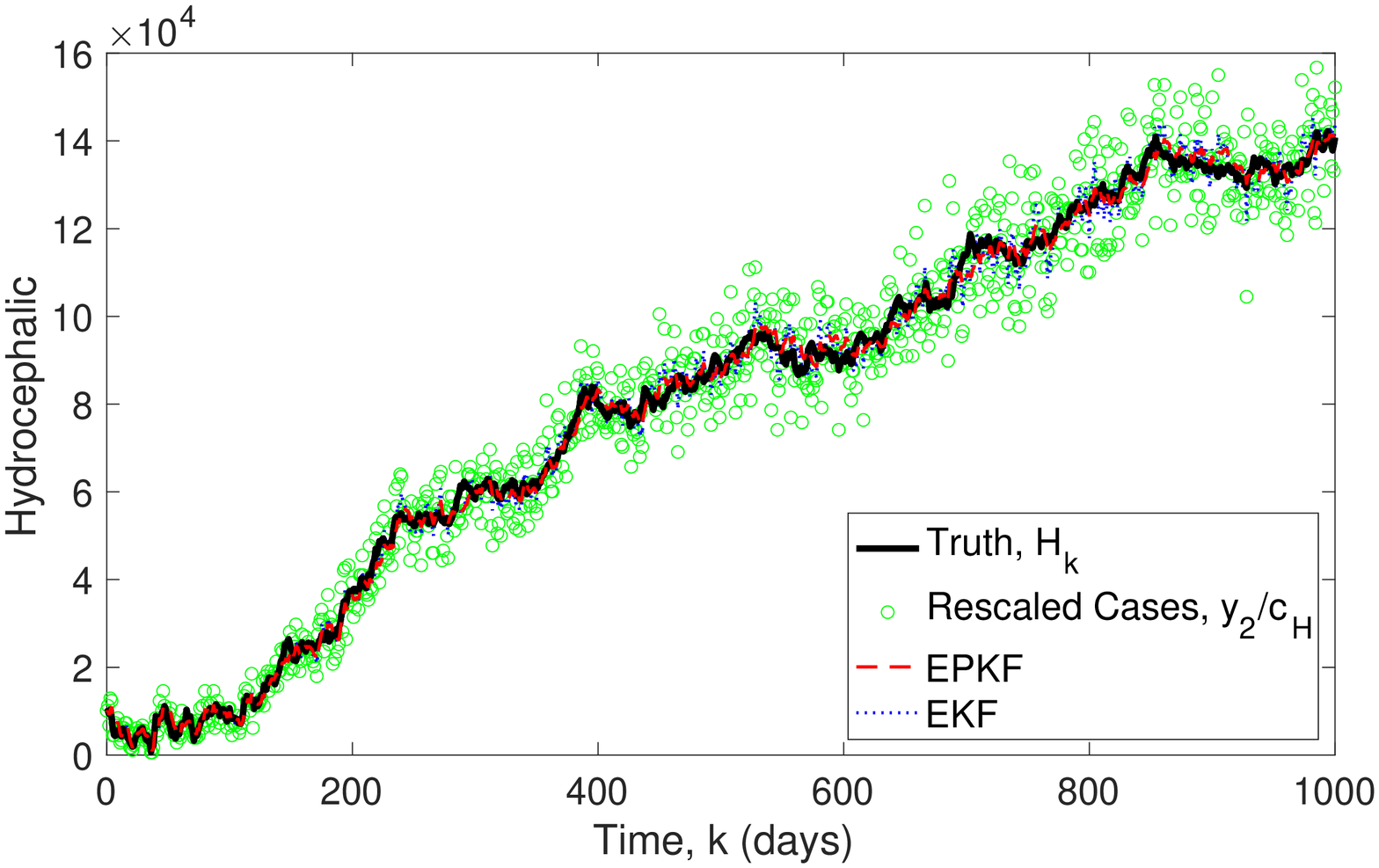}} \vspace{-4pt} \\
    \subfigure[]{\includegraphics[width=.40\linewidth]{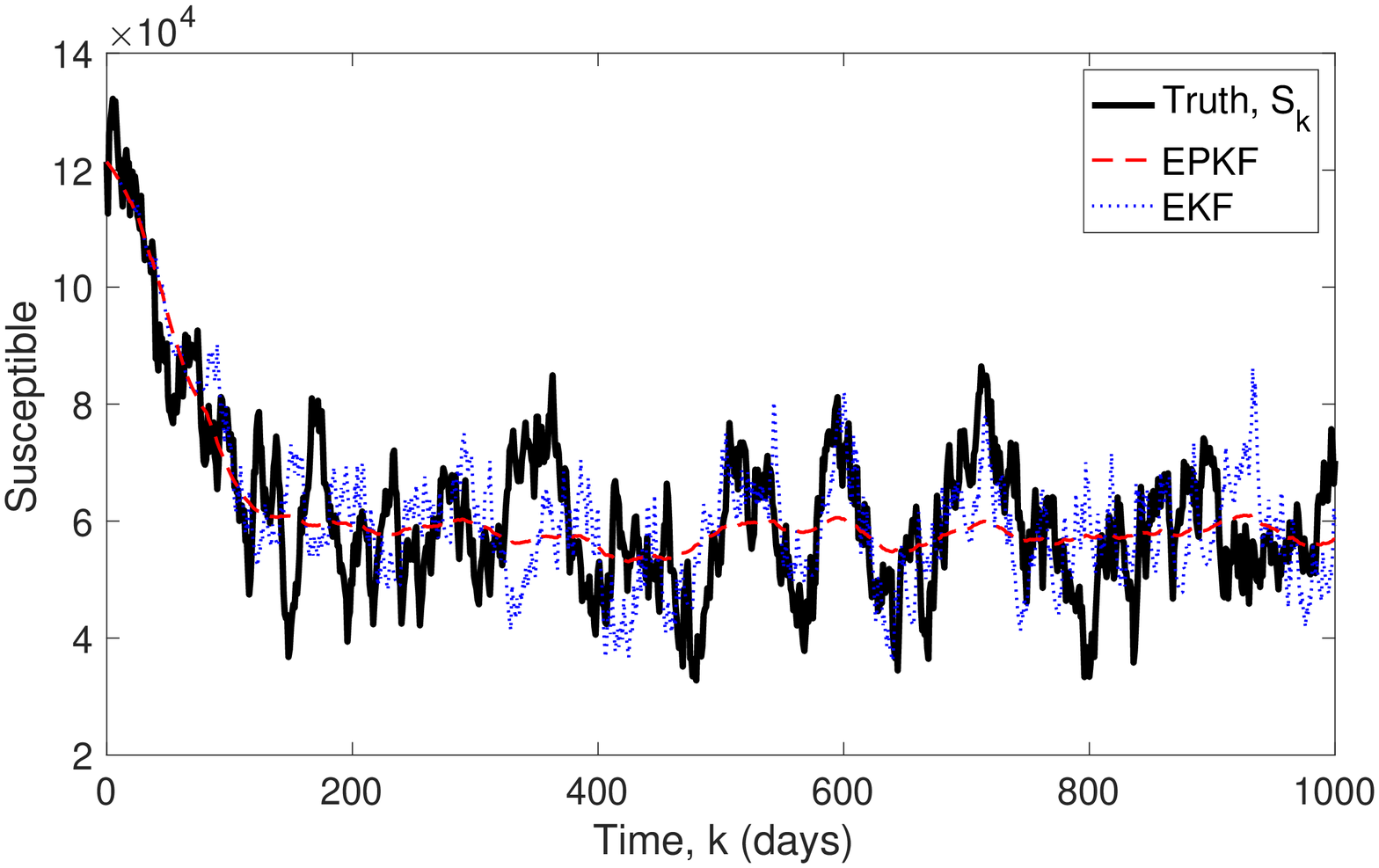}}
    \subfigure[]{\includegraphics[width=.40\linewidth]{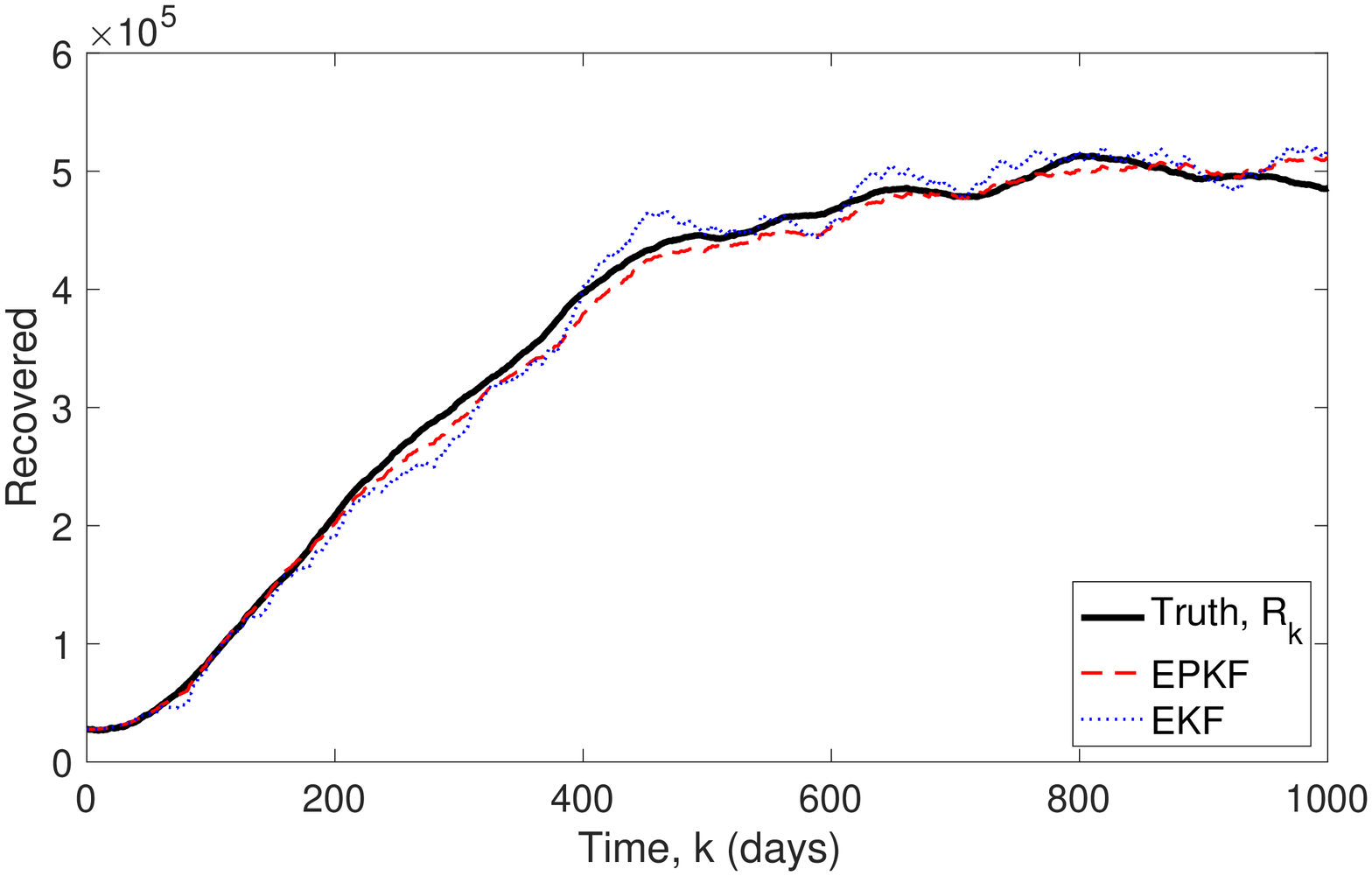}} \vspace{-4pt} \\
    \subfigure[]{\includegraphics[width=.40\linewidth]{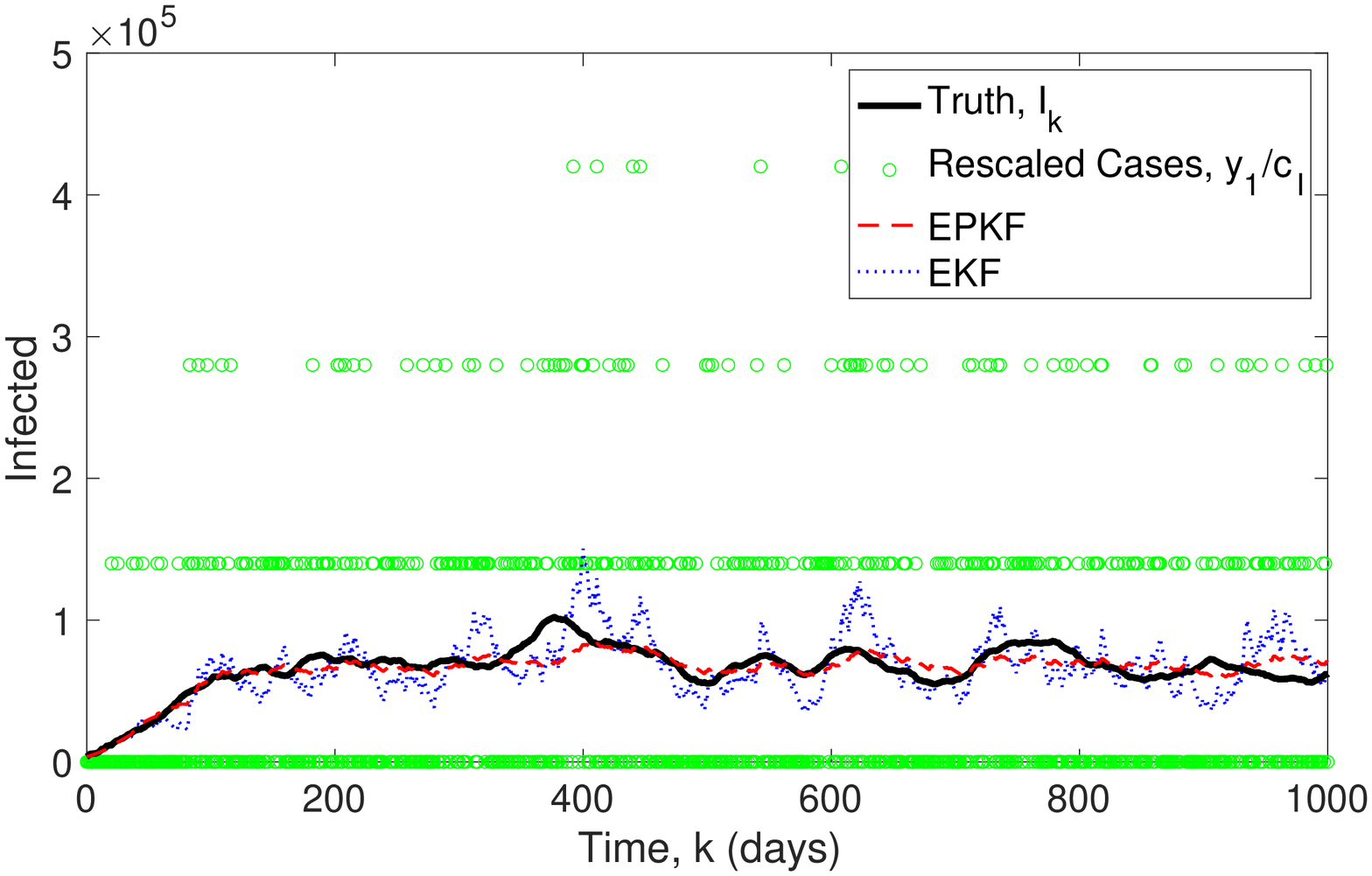}}
    \subfigure[]{\includegraphics[width=.40\linewidth]{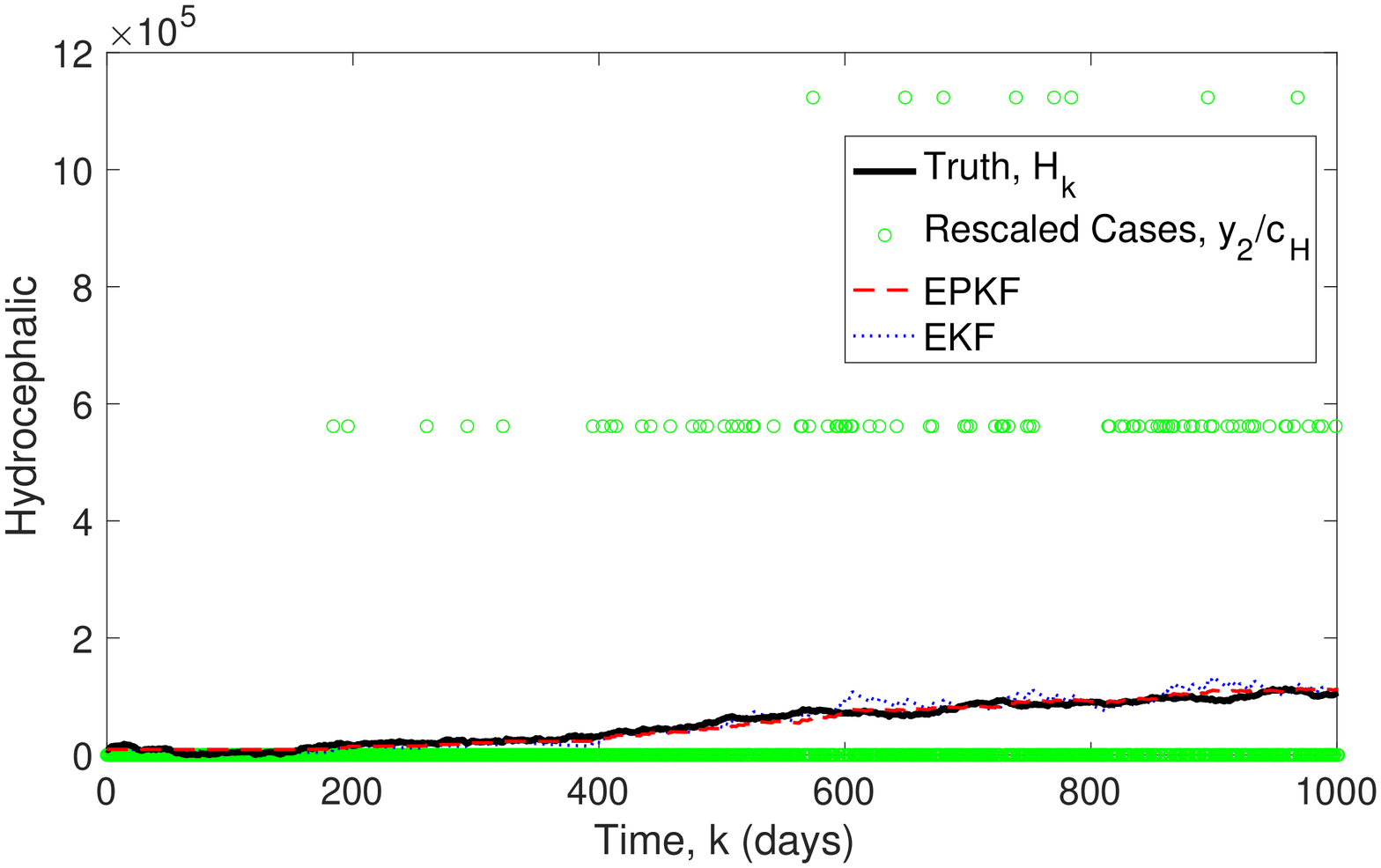}}\vspace{-4.6pt}
    \caption{Comparison of the Extended PKF (optimal variable gain) and the Extended Kalman filter (optimal fixed gain for the noncontagious equilibrium) for the contagious SIRH model with Poisson observations of the infected and hydrocephalic populations. We compare the true S, I, R and H values (black) to the PKF (red, dashed) and Kalman filter (blue, dotted) estimates.  (a)-(d) show a standard observation rate, $c_I=0.2/T_S, c_H = 0.6/T_R$, and (e)-(h) show a low observation rate, $c_I=0.0002/T_S, c_H = 0.0006/T_R$. Infected and hydrocephalic also show the observations (green, circles) rescaled by dividing by the constants $c_I, c_H$ respectively. The system is initialized at the noncontagious equilibrium and run forward with $\beta = 10^{-6}$ simulating the introduction of a contagious source of infection which moves the system to a new equilibrium.}
    \label{estim1FigContagious}
\end{figure}

A standard method for lifting the Kalman filter to the nonlinear setting is the Extended Kalman Filter (EKF) \cite{simon2006optimal}.  The EKF uses the nonlinear dynamics to produce the forecast $x_{k+1}^- = f_k(x_{k}^+)$, and a linear approximation to the dynamics is used for forecasting the covariance matrix $P_{k+1}^- = F_k P_{k}^+ F_k^\top + W$.  To define $F_k$ the EKF linearizes the dynamics around the current state estimate, setting $F_k = Df_k(\hat x_{k}^+)$.  This approximates the nonlinear dynamics as a nonautonomous linear system for the purposes of forecasting the covariance estimates.  In the example below we apply the EKF using
\begin{equation}
    F_k = Df_k(\hat x_{k}^+) = \left( \begin{array}{cccc} 1-d-a-g_S-\beta I_{k}^+ & -\beta S_{k}^+ & 0 & 0 \\ a+\beta I_{k}^+ & 1-d-d_I-c + \beta S_{k}^+ & 0 & 0 \\ 0 & c & 1-d_R-g_R-h & 0 \\ 0 & 0 & h & 1-d_R - d_H \end{array} \right).  \nonumber
\end{equation}
Since the PKF is also based on the Kalman equations, we can use this same idea to extend the PKF to nonlinear systems which we call the Extended PKF (EPFK).

\begin{figure}
    \centering
    \subfigure[]{\includegraphics[width=.49\linewidth]{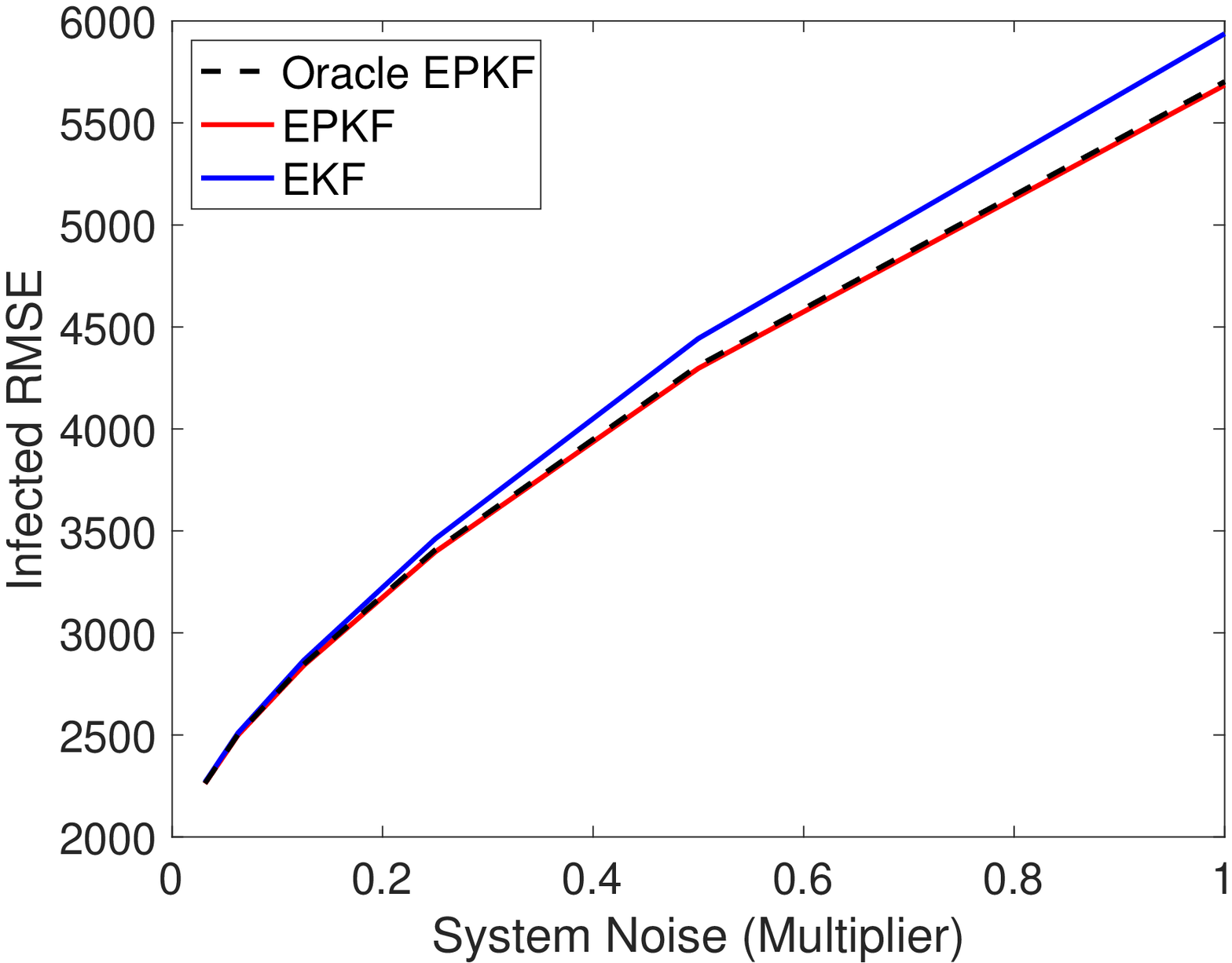}}
    \subfigure[]{\includegraphics[width=.49\linewidth]{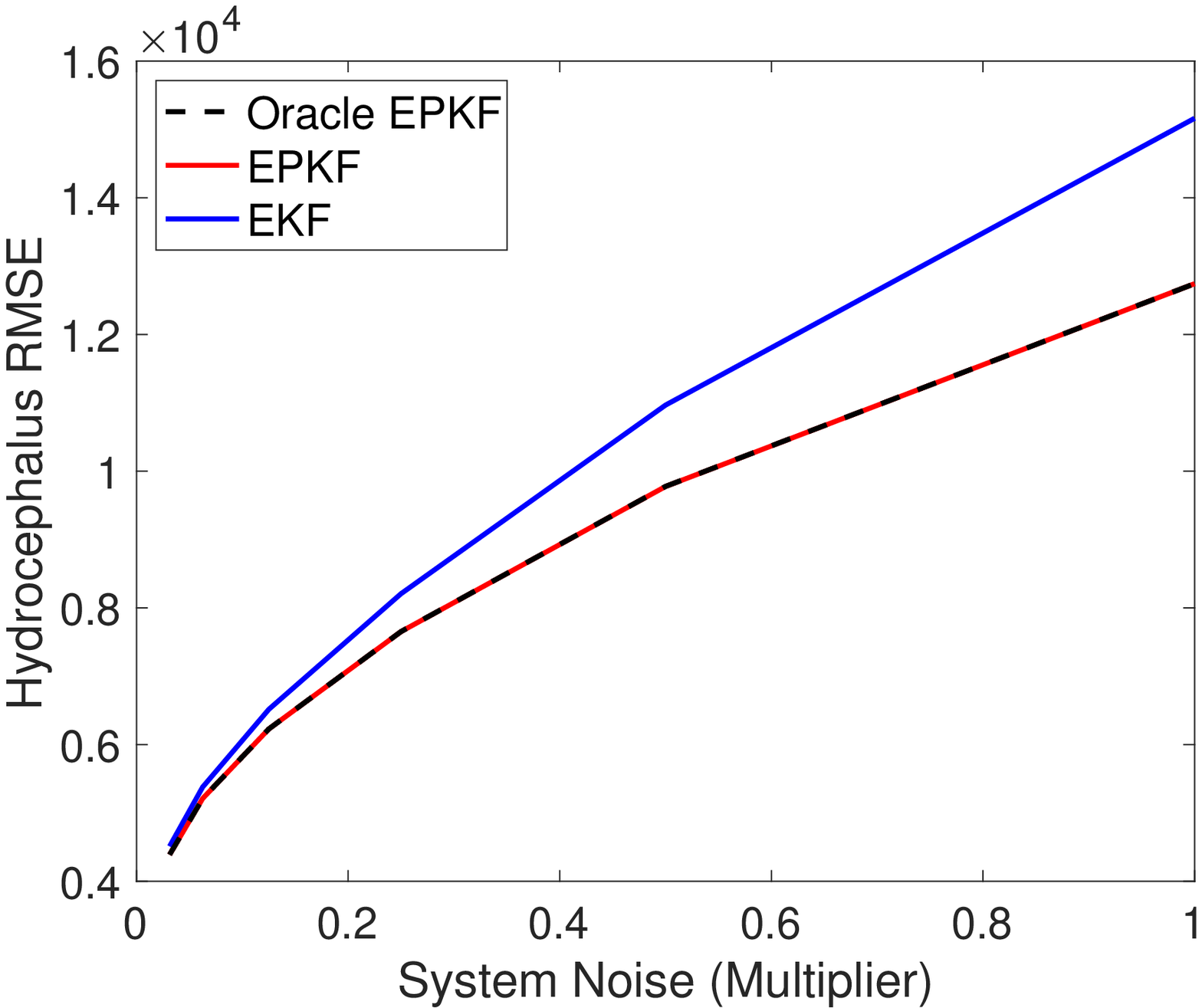}}
    \subfigure[]{\includegraphics[width=.46\linewidth]{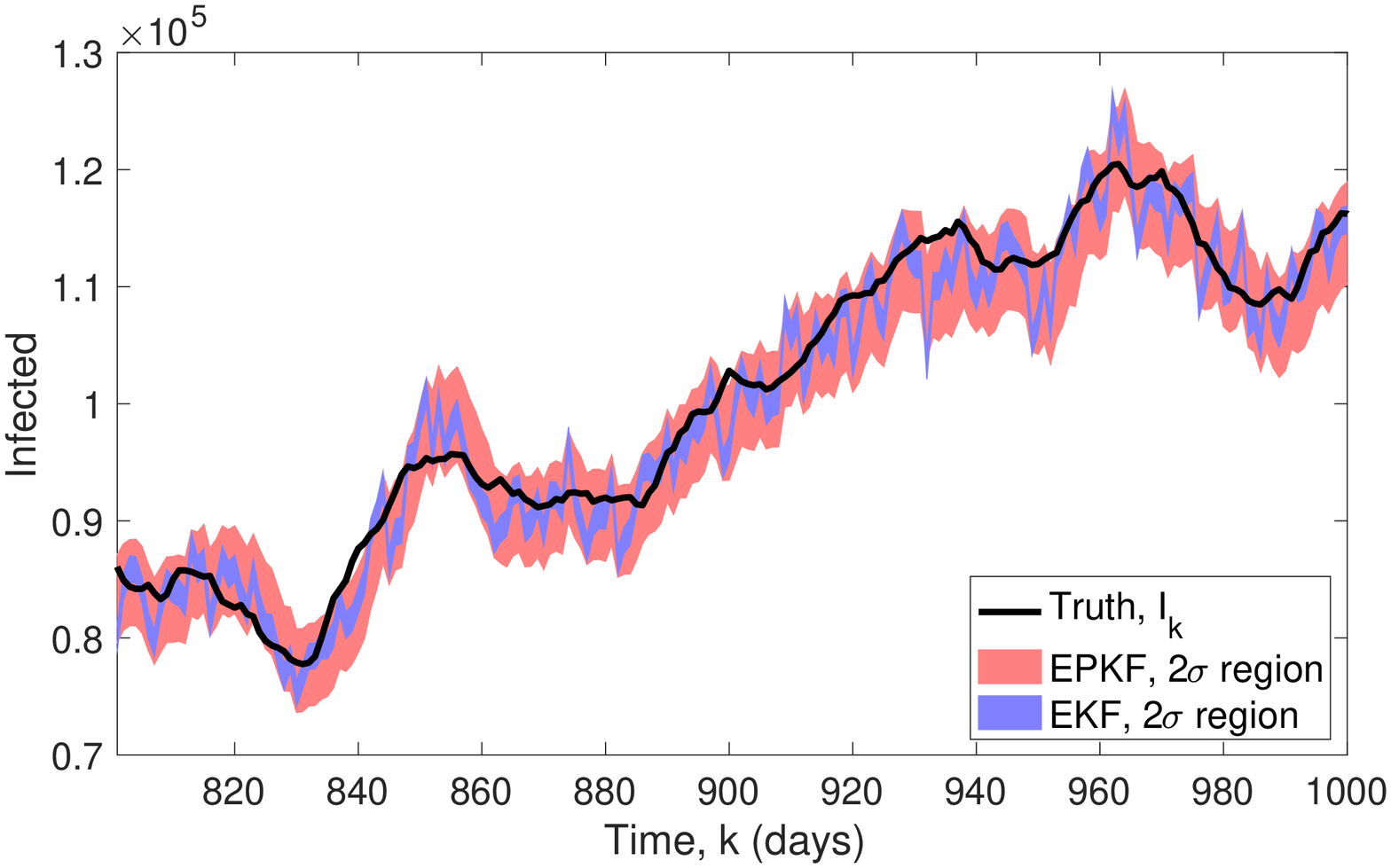}}
    \subfigure[]{\includegraphics[width=.46\linewidth]{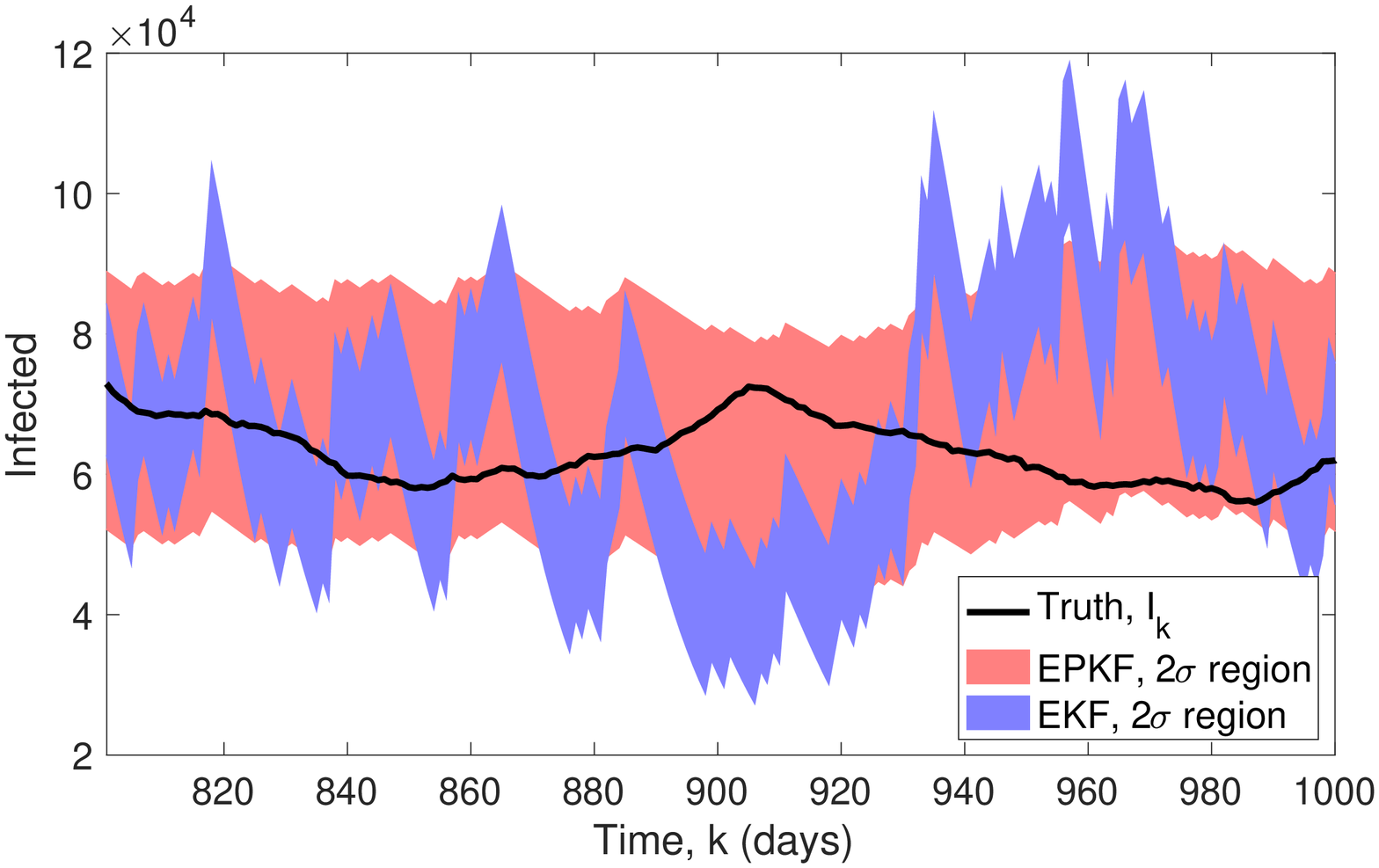}}
    \caption{(a)-(b) Comparison of the RMSE of the EPKF (red, optimal variable gain) and the EKF (blue, optimal fixed gain) as function of the system noise. We also compare to an oracle EPKF (black, dashed) which is given the optimal choice of $V_k = \textup{diag}(B\vec x_k)$.  System noise is quantified as a multiple of the base noise level $W$. The RMSE is averaged over $10^6$ filter steps. (c)-(d) The two standard deviation region around the EPKF and EKF estimates are shown using the variance estimated by the respective filters for the standard observation rate (c) and the low rate (d) from Fig.~\ref{estim1FigContagious}.  }
    \label{EPKFrmseComparison}
\end{figure}

In Fig.~\ref{estim1FigContagious} we simulate the system \eqref{SIRHcontagious} with $\beta = 10^{-6}$ initialized at the noncontagious equilibrium found in Section \ref{SIRHmodel} (all other parameters are the same as in Section \ref{SIRHmodel}).  This simulates the introduction of a contagious source of disease to a system that had stabilized at the noncontagious equilibrium.  To demonstrate the advantage of the EPKF over the non-Poisson version, we assume that the EKF is given the optimal fixed gain for the noncontagious equilibrium. Fig.~\ref{estim1FigContagious} shows that as the number of infected and hydrocephalic cases increase the EKF estimate becomes very noisy, since it is underestimating the observation variance and as a result follows the observations too closely.  This shows how the EPKF is able to automatically adapt to the new equilibrium.  The top tow rows of Fig.~\ref{estim1FigContagious} show the simulation with the standard observation rate, $c_I=0.2/T_S, c_H = 0.6/T_R$.  To demonstrate the ability of the EPKF to assimilate at very low observation rates, we repeated the experiment after reducing the observation rates by a factor of $1000$, and these results are shown in the bottom two rows of Fig.~\ref{estim1FigContagious}.  In this case, due to the low observation rates, of the 1000 days shown in Fig.~\ref{estim1FigContagious}, 624 days had zero infected reported and 908 days had zero hydrocephalic reported. Despite these large numbers of zeros, the EPKF effectively assimilates the available information.  

Another critical aspect of filter performance is the accuracy of the uncertainty quantification provided by the covariance matrix.  In the bottom row of Fig.~\ref{EPKFrmseComparison} we compare the true signal (black) to the two standard deviation region around the EPKF and EKF estimates.  These regions are generated by adding and subtracting twice the square root of the diagonal entry of the covariance matrix estimate at each time step.  An accurate uncertainty quantification would imply that the truth only leaves the region for around 10 of the 200 time steps shown. Notice that the EKF significantly underestimates the variance, meaning that it is overconfident in its estimator, because the true signal leaves the region much more than expected.  The EPKF gives a more accurate uncertainty quantification, perhaps slightly overestimating the uncertainty in the low observation rate case since the true signal never leaves the region.

We also compare the EPKF to the standard EKF using the optimal fixed gain for the contagious equilibrium, starting from the contagious equilibrium.  Fig.~\ref{EPKFrmseComparison} shows that the EPKF has the largest advantage at high system noise levels (as in Fig.~\ref{estim1Fig2}) due to the absence of unstable directions in the deterministic dynamics near equilibrium.  This suggests that the EPKF would have an even more significant advantage for chaotic systems.  Also, as in Fig.~\ref{estim1Fig2}, we compare the empirical EPKF, which uses the state estimate to determine the observation variance, to an oracle version of the EPKF that uses the true state to determine the observation variance, and again the performance is very similar.  

Finally, we note that two closely related alternative approaches to applying the Kalman filter to nonlinear dynamics are the Ensemble Kalman Filter (EnKF) \cite{simon2006optimal} and Ensemble Adjustment Kalman Filter \cite{anderson2001ensemble}. Both of these methods use an ensemble forecast instead of linearizing the dynamics to estimate $P_{k+1}^-$.  Since the EnKF and EAKF are also based on the Kalman formulas the PKF method can be applied just as easily to these methods. A potentially significant issue for most ensemble methods is that ensemble members are typically not guaranteed to be positive, leading to negative populations in the ensemble members.  One possible method to address this would be to log-transform the model, however this has the downside of introducing stronger nonlinearities into the model. A related idea would be to replace the observations with the variance-stabilizing transformation $\sqrt{y_k + 1/4}$ which is approximately Gaussian with mean $\sqrt{Bx_k}$ and constant variance $1/4$.  Both these approaches simplify the statistics at the cost introducing additional nonlinearity into the model or observation function, and exploring these tradeoffs is an interesting subject for future investigations.

This approach is closely related to the recent work of \cite{li2020substantial} applied to COVID-19 with Poisson observations. There, an EAKF was used with a heuristically chosen observation covariance $V_k$ which was proportional to the square of the observations.  In fact, \cite{li2020substantial} also suggested an alternative of using $V_k$ proportional to the observations.  Our analysis of the PKF shows that in fact the optimal choice for linear dynamics is to set $V_k$ equal to the \emph{predicted} observations, as we propose in the EPKF.  In fact, our analysis of a nonlinear contagious model in this section suggests that \cite{li2020substantial} were very close to the optimal approach.

\section{Conclusion}\label{futuredirections}

The mathematical methods of filtering and control originated with linear models, direct observations, and Gaussian noise. However, these assumption may not be appropriate in the context of disease modeling, where the observation of cases of communicable and noncommunicable disease often present as Poisson processes. Unfortunately, the customary Kalman filter is not well suited to assimilate such Poisson occurrences and estimate the true number of underlying cases. The Poisson Kalman Filter (PKF) is an optimal filter for such surveillance. 

The linear PKF is a very general filter suitable for a broad range of noncommunicable disease observations where the nonlinear interaction of susceptible and diseased individuals is not an inherent component of disease initiation (including noninfectious disease such as diabetes or stroke). We extended our findings to encompass the nonlinear interactions of susceptible and infected individuals typical of contagious disease through an extended PKF or EPKF.

We also created, to our knowledge, the first SIRH compartmental model that can be used in the surveillance of neonatal sepsis and postinfectious hydrocephalus, endemic disease that causes tremendous numbers of yearly global deaths in the developing world. In particular, we incorporated both the noncommunicable and communicable dynamics that have been observed in these infant infections.

Additionally, our case study of sepsis and hydrocephalus suggests many promising directions for future development. If a more careful tracking of cases is desired, the neonatal and infancy periods can be segmented in to multiple stages. For example, it is well known that the infections that are acquired perinatally from the mother, so called early onset sepsis, are manifest within the days of the first week of life. Infections during the subsequent weeks of the neonatal period (first 4 weeks) are environmentally acquired and are typically a very different spectrum of organisms. Therefore $S_{(0i)}$ could represent susceptible at $(0-i)$-days after birth, and $S_{(ij)}$ could represent susceptibles from $(i-j)$-days after birth, with varying rates and risks from sepsis at different stages of development.  Another critical factor in sepsis and hydrocephalus cases is environmental variables such as rainfall \cite{rainfall2012}, which suggests that a full spatiotemporal model will be necessary to more fully represent these dynamics. Recent findings \cite{Paulson2020} demonstrate that more than one infection (co-infection) can be found in some of these infants -- perhaps even a mixture of noncommunicable bacteria and communicable viruses -- demonstrating that a mixed linear and nonlinear model would be required to represent such co-infections.  A spatiotemporal model would allow the optimal control to consider multiple methods of control and determine ideal locations and times to apply each.  

The recent coronavirus (COVID-19) epidemic is one where Poisson dynamics are required in the modeling and data assimilation \cite{li2020substantial}. In this article, we derive the Kalman equations that lead to the optimal linear filter, and propose that the optimal choice is to set the observed covariance equal to the predicted observations as proposed in the nonlinear EPKF.

\section{Acknowledgements}\label{acknowledgements}
We are grateful to D. Simon, M. Ferrari and M. Norton for their helpful discussions. Funded by an NIH Director's Transformative Award 1R01AI145057.

\bibliographystyle{apsrev4-2}
\bibliography{SIRH}

\appendix 

\section{Optimal Linear Filter Derivation}\label{filterapp}

We develop a recursive weighted least square (WLS) estimator by determining how to estimate a constant on the basis of several measurements that follow a Poisson distribution. We then use it as a basis for developing a discrete-time Kalman filter that uses Poisson distributed measurements, which we accomplish by adapting the techniques used in developing a discrete-time Kalman filter~\cite{simon2006optimal}. We will remove the arrow symbols from vectors in this section to reduce the number of symbols in formulas. 

\subsection{Recursive weighted least squares estimator with Poisson observations}
The measurement equation of a linear stochastic discrete-time dynamic system  with indirect measurements of the state is given as
\begin{align}
x_k &=   x_{k-1}   \nonumber \\
y_k & = z_k, \qquad \qquad \qquad    z_k \sim \mbox{Poisson}(\lambda_k = B  x_k) \label{eqn_rlsdyn}
\end{align}
where $x_k \in \mathbb{R}^n$, $y_k \in \mathbb{R}^m$ are the state vector and
measurement vector respectively, and $B \in \mathbb{R}^{m \times n}$ is a known deterministic matrix. The measurement variable $z_k \in \mathbb{R}^m$ follows a random Poisson distribution 
\begin{align}
p(z_k | x_k) = \frac{(B x_k)^{z_k}}{z_k !} e^{-(B x_k)}, \qquad z = 0, 1, 2, \ldots \label{eqn_poss_dist}
\end{align}
where the rate $\lambda_k = B x_k \geq 0$. We note that the positive real rate $\lambda_k$ equals the expected value $\mathbb{E}[z_k]$ and variance $\textup{Var}(z_k)$ such that $$\mathbb{E}[z_k] = \textup{Var}(z_k) = \lambda_k  $$ 

 In our case, we assume that the state $x$ is non-negative. We also assume that each Poisson random variable $(z_k)_i$ has a rate that is only dependent on the corresponding state $(x_k)_i$, such that $\mathbb{E}[(z_k)_i] = c_i (x_k)_i$ where $c_i \in \mathbb{R}_+$ are known deterministic parameters. A linear recursive estimator can be written as
\begin{align}
y_k &=  z_k   \\
\hat{x}_k &= \hat{x}_{k-1} + K_k \left(  y_k -  B \hat{x}_{k-1} \right)\label{eqn_rls}
\end{align}
where $K_k$ is the optimal gain matrix to be determined.
\begin{thm}
\label{theorem_unbiased}
The estimator of~(\ref{eqn_rls}) is an unbiased estimator of $x_k$; that is, $\mathbb{E}[ x_k - \hat{x}_{k}] = 0 $
\end{thm}
\begin{proof}
Calculating the estimation error mean, we write
\begin{align}
\mathbb{E}[\epsilon_{x,k}] &= \mathbb{E}[x_k - \hat{x}_k]\nonumber \\
&= \mathbb{E}\left[ x_k - \hat{x}_{k-1} - K_k \left(  y_k -  B \hat{x}_{k-1} \right) \right]\nonumber \\
&= \mathbb{E}\left[ \epsilon_{x,k-1} - K_k \left( z_k -  B \hat{x}_{k-1} \right) \right]\nonumber \\
&= \mathbb{E}\left[ \epsilon_{x,k-1} - K_k \left( z_k - B x_{k-1}  + B x_{k-1}  -  B \hat{x}_{k-1} \right) \right]\nonumber \\
&= \mathbb{E}\left[ \epsilon_{x,k-1} - K_k B \epsilon_{x,k-1} - K_k \left(z_k - B x_{k-1}  \right) \right]\nonumber \\
&=  \mathbb{E} \left[ ( I - K_k B ) \epsilon_{x,k-1}  - K_k  ( z_k  - B x_{k-1}  )  \right] \nonumber \\
&= \left( I - K_k B \right) \mathbb{E} \left[ \epsilon_{x,k-1} \right] - K_k \left(\mathbb{E}[z_k] - B x_{k-1} \right)  \label{eqn_rls_errM}
\end{align}
So since $\mathbb{E}[z_k] = B x_{k} = B x_{k-1}$ and inductively we assume $\mathbb{E}\left[ \epsilon_{x,k-1} \right] = 0$, we have $\mathbb{E}[\epsilon_{x,k}] = 0$. Therefore~(\ref{eqn_rls}) is an unbiased estimator. 
\end{proof}

Note that the unbiased estimator property holds regardless of the value of the gain matrix $K_k$. This implies that, on average, the state estimate $\hat{x}_k$ will be equal to the true state $x_k$, when measurements -- that follow a Poisson distribution whose rate is dependent on the state -- are taken. 
Moreover, we note that Theorem \ref{theorem_unbiased} holds whenever $\mathbb{E}[z_k] = B x_k$, so as long as the expected value of the observations is linear in the state one could choose an appropriate $B$ to have an unbiased estimator.  We now turn to the construction of the optimal linear filter. 

\begin{thm}\label{minimumvariance_thm}
Among all linear filters, the filter given by the linear estimator of~\ref{eqn_rls} with gain matrix $K_k$ given by
\[ K_k  = P_{k-1} B^T  \left( B P_{k-1} B^T    + V_k \right)^{-1} \] is optimal in the sense of minimal sum of squared errors when $V_k = \textup{diag}(B x_k)$. In other words, \[\frac{\partial J_k}{\partial K_k} = 0 \]
where 
\[J_k =  \textup{trace}(P_k) =  \mathbb{E}[||\hat x_k -   x_k||_2^2] = \sum_i \mathbb{E}[(\hat x_k -  x_k)_i^2] \] 
\end{thm}
\begin{proof}
Using~(\ref{eqn_rls_errM}), we solve for the estimation error covariance $P_k$ as
\begin{align}
P_k  = & \mathbb{E}\left[\epsilon_{x,k} \epsilon_{x,k}^T   \right] \nonumber \\
=& \mathbb{E} \left\lbrace \left[ ( I - K_k B  )  \epsilon_{x,k-1}  - K_k ( z_k -  B x_{k-1} ) \right] \left[   \epsilon_{x,k-1}^T( I - K_k B )^T    - (z_k -  B x_{k-1} )^T K_k^T  \right]^T   \right\rbrace \nonumber \\
=& ( I - K_k B) \mathbb{E}[\epsilon_{x,k-1} \epsilon_{x,k-1}^T]( I - K_k B )^T -  ( I - K_k B) \mathbb{E}[\epsilon_{x,k-1} (z_k -  B x_{k-1} )^T] K_k^T  \nonumber \\
& - K_k \mathbb{E}[(z_k -  B x_{k-1} )\epsilon_{x,k-1}^T]( I - K_k B )^T +  K_k \mathbb{E}[(z_k -  B x_{k-1} )    (z_k -  B x_{k-1} )^T] K_k^T \label{eqn_pk1}
\end{align}
Since the estimation error at time $k-1$ given by $\epsilon_{x,k-1} = x - \hat x_{k-1}$ is independent of the measurement $z_k$ at time $k$, we have 
\begin{eqnarray*}
\mathbb{E}[\epsilon_{x,k-1} (z_k -  B x_{k-1} )^T]  &=& \mathbb{E}[\epsilon_{x,k-1}]\mathbb{E}[(z_k -  B x_{k-1} )^T]   \\
&=& 0  
\end{eqnarray*}
since the expected value $\mathbb{E}[\epsilon_{x,k-1}]$ and $\mathbb{E}[ z_k -  B x_{k-1}]$ are both zero.
More generally, when $x_{k-1} = x_k$ is a random variable, $\epsilon_{x,k-1}$ may not be independent of the measurement $z_k$, however, the above expectation is still zero since,
\begin{eqnarray*}
\mathbb{E}[\epsilon_{x,k-1} (z_k -  B x )^T]  &=& \mathbb{E}[(x_{k-1} - \hat x_{k-1} ) ( z_k -  B x_{k-1} )^T] \\ &=& \mathbb{E}[\mathbb{E}[(x - \hat x_{k-1} ) (z_k -  B x_{k-1} )^T \, | \, x_{k-1}, z_1,...,z_{k-1}]] \\ &=& \mathbb{E}[(x_{k-1} - \hat x_{k-1} ) \mathbb{E}[( z_k -  B x_{k-1} )^T \, | \, x_{k-1},z_1,...,z_{k-1}]] \\ &=& \mathbb{E}[(x_{k-1} - \hat x_{k-1} ) (\mathbb{E}[z_k \, | \, x_{k-1},z_1,...,z_{k-1}] -  B x_{k-1} )^T ] \\ &=& \mathbb{E}[(x_{k-1} - \hat x_{k-1} ) (\mathbb{E}[z_k \, | \, x_{k-1}] -  B x_{k-1} )^T ]
\\ &=& \mathbb{E}[(x_{k-1} - \hat x_{k-1} ) (Bx_{k-1} -  B x_{k-1} )^T ] \\ 
&=& 0  
\end{eqnarray*}
where the second equality follows from the law of total expectation and $\mathbb{E}[z_k \, | \, x_{k-1},z_1,...,z_{k-1}] = \mathbb{E}[z_k \, | \, x_{k-1}] = Bx_{k-1}$ since the random variables $z_1,...,z_k$ are conditionally independent given $x_{k-1}=x_k$.
Therefore,~(\ref{eqn_pk1}) reduces to
\begin{align}
P_k =& ( I - K_k B) \mathbb{E}[\epsilon_{x,k-1} \epsilon_{x,k-1}^T]( I - K_k B )^T +  K_k \mathbb{E}[(z_k -  B x_{k-1} ) (z_k -  B x_{k-1} )^T] K_k^T \label{eqn_pk2}
\end{align}
Using the fact that $Bx = \mathbb{E}[z_k]$, we rewrite~(\ref{eqn_pk2}) as
\begin{align}
P_k =& ( I - K_k B)\mathbb{E}[\epsilon_{x,k-1} \epsilon_{x,k-1}^T]( I - K_k B )^T +  K_k \mathbb{E}[(z_k -  \mathbb{E}[z_k] ) (z_k - \mathbb{E}[z_k] )^T] K_k^T \label{eqn_pk3}
\end{align}
Recall that for a random variable $Y$ with mean $\mathbb{E}[Y]$, the $i$th central moment of $Y$, which is written as $$ i\mbox{th central moment of } Y =   \mathbb{E}[(Y - \mathbb{E}[Y])^i]$$
equals its variance when $i = 2$ (see Chapter 2 of~\cite{simon2006optimal}). Therefore, 
\begin{align}
\mathbb{E}[(z_k -  \mathbb{E}[z_k] ) ( z_k - \mathbb{E}[ z_k] )^T] = V_k \label{eqn_varZ}
\end{align}
where $V_k \in \mathbb{R}^{m \times m}$, which is written as $V_k = \textup{diag}(B x_k)$, is the covariance of $z_k$. Substituting~(\ref{eqn_varZ}) into~(\ref{eqn_pk3}) gives
\begin{align}
P_k =& ( I - K_k B) P_{k-1} ( I - K_k B )^T +  K_k V_k K_k^T \label{eqn_pk4}
\end{align}
which is the recursive formula for determining the covariance of the least squares estimation error. We then minimize the sum of the estimation error variances at time $k$. From the cost function
\begin{equation}
J_k = \textup{trace}(P_k) \label{eqn_rls_cost}
\end{equation}
we write 
\begin{equation}
\frac{\partial J_k}{\partial K_k} = 2(I - K_k B) P_{k-1}(-B)^T + 2K_k V_k   \label{eqn_rls_cost1}
\end{equation}
Setting~(\ref{eqn_rls_cost1}) equal to zero to find the value of $K_k$ that minimizes $J_k$,
\begin{eqnarray}
K_k V_k  & = & (I - K_k B) P_{k-1} B^T  \nonumber \\
K_k \left( V_k +  B P_{k-1} B^T  \right) & =& P_{k-1} B^T  \nonumber \\
K_k  & =& P_{k-1} B^T  \left( B P_{k-1} B^T    + V_k \right)^{-1} \label{eqn_kalmanGain}
\end{eqnarray}
This implies that the optimal gain matrix $K_k$ given by~(\ref{eqn_kalmanGain}) minimizes the sum of squared errors when $V_k = \textup{diag}(B x_k)$.

\end{proof}
We note that if all the states $x$ are used to generate the output $y$, such that each state $x_i$ is used to generate its Poisson random measurement $z_i$, then $V_k = \mbox{diag}([c_1 x_{1_k} , \:\:  c_2 x_{2_k},  \:\: \cdots, \:\: c_m x_{m_k} ] ) $ where $m = n$. Since the true state $x_k$ is unavailable to the estimator, we replace $x_k$ with $\hat{x}_{k-1}$. Therefore we have $$V_k = \mbox{diag}([c_1 \hat{x}_{1_{k-1}} , \:\:  c_2 \hat{x}_{2_{k-1}},  \:\: \cdots, \:\: c_m \hat{x}_{m_{k-1}} ] )  = \mbox{diag}(B \hat{x}_{k-1} )$$
which results in a suboptimal filter.

\subsection*{Recursive weighted least square estimator algorithm}
 \begin{enumerate}
 \item[1] Initialization
 \begin{eqnarray}
 \hat{x}_0 &=& \mathbb{E}[x] \nonumber \\
 P_0 &=& \mathbb{E} \left[(x - \hat{x}_0)(x - \hat{x}_0)^T \right]\nonumber
 \end{eqnarray}

 \item[2] Estimation
 \begin{eqnarray}
K_k & = & P_{k-1} B^T  \left( B P_{k-1} B^T    + V_k \right)^{-1}  \label{eqn_rls1}\\
\hat{x}_k &=& \hat{x}_{k-1} + K_k \left(  y_k - B \hat{x}_{k-1}   \right) \label{eqn_rls2} \\
P_k &=&  ( I - K_k B) P_{k-1} ( I - K_k B )^T +  K_k V_k K_k^T  \label{eqn_rls3}
 \end{eqnarray}
\end{enumerate}

\subsection{Kalman filter based on WLS with Poisson observations}
Consider the linear stochastic discrete-time dynamic system  with indirect measurements of the state given by 
\begin{align}
x_k &= F_{k-1} x_{k-1} +  G_{k-1}u_{k-1} + w_{k-1},   \qquad   w_k \sim \mathcal{N}(0,\,\sigma^{2}) \label{eqn_KFdyn} \\
y_k & = z_k, \quad \quad \qquad \qquad \qquad \qquad \qquad \qquad    z_k \sim \mbox{Poisson}(\lambda_k = B  x_k) \label{eqn_KFdyn1}
\end{align}
The expected value of both sides of~(\ref{eqn_KFdyn}) is given as
\begin{align}
\bar{x}_k = \mathbb{E}[x_k] = F_{k-1} \bar{x}_{k-1} +  G_{k-1}u_{k-1}
\end{align}
Using
\begin{align}
(x_k - \bar{x}_k)(x_k - \bar{x}_k)^T = & F_{k-1} (x_{k-1} - \bar{x}_{k-1})(x_{k-1} - \bar{x}_{k-1})^T F_{k-1}^T +  w_{k-1} w_{k-1}^T\nonumber \\
 & + F_{k-1} (x_{k-1} - \bar{x}_{k-1})w_{k-1}^T  + w_{k-1}(x_{k-1} - \bar{x}_{k-1})^T F_{k-1}^T
\end{align}
the covariance of $x_k$ is given as
\begin{align}
P_k^- = & \mathbb{E} \left[ (x_k - \bar{x}_k)(x_k - \bar{x}_k)^T   \right] \nonumber \\
= & F_{k-1} P_{k-1}^+ F_{k-1}^T + W_{k-1}\label{eqn_KF_pk}
\end{align}
because $\mathbb{E}[(x_{k-1} - \bar{x}_{k-1})w_{k-1}^T] = 0$, since $(x_{k-1} - \bar{x}_{k-1})$ is uncorrelated with $w_{k-1}$.  Therefore from~(\ref{eqn_KF_pk}), (\ref{eqn_rls1}), (\ref{eqn_rls2}), and (\ref{eqn_rls3}), we replace $\hat{x}_{k-1}$ with $\hat{x}_k^-$,  we replace $P_{k-1}$ with $P_k^-$, we replace $\hat{x}_k$ with $\hat{x}_k^+$, and we replace $P_k$ with $P_k^+$. We then get the Poisson Kalman filter equations for each time step $k = 1, 2, \cdots: $
\begin{eqnarray}
\hat{x}_k^+ &=& \hat{x}_k^- + K_k \left(y_k - B \hat{x}_k^-   \right) \\
P_k^- &= & F_{k-1} P_{k-1}^+ F_{k-1}^T + W_{k-1} \\
K_k & = & P_k^- B^T  \left( B P_k^- B^T    + V_k \right)^{-1} \\
\hat{x}_k^- &=&  F_{k-1} \hat{x}_{k-1}^+  +  G_{k-1}u_{k-1} \\
P_k^+ &=&  ( I - K_k B) P_k^- ( I - K_k B )^T +  K_k V_k K_k^T \label{eqn_covP} 
\end{eqnarray}

\end{document}